\documentclass[letterpaper, 10pt, conference]{ieeeconf}

\usepackage[active]{srcltx}

\usepackage{epsfig}
\usepackage{amsmath}
\usepackage{amssymb}
\usepackage{psfrag}
\usepackage{stackrel}
\usepackage{graphicx}

\newcommand{\eps}{\varepsilon}

\newtheorem{define}{Definition}

\newtheorem{theo}[define]{Theorem}

\newtheorem{lem}[define]{Lemma}

\newcommand{\N}{\mathbb{N}}
\newcommand{\Z}{\mathbb{Z}}
\newcommand{\R}{\mathbb{R}}

\title{ \LARGE \bf A Stochastic Markov Model for Coordinated Molecular Motors }
\author{
	\authorblockN{
	Donatello Materassi \authorrefmark{1},
	Subhrajit Roychowdhury \authorrefmark{1},
	and Murti V.~Salapaka\authorrefmark{1}, %
	\authorblockA{
		\authorrefmark{1}
		Department of Electrical and Computer Engineering,\\
		University of Minnesota,\\
		200 Union St SE,\\
		55455, Minneapolis (MN).\\
		{\tt\small \{mater013|subhra|murtis\}@umn.edu}}\\
	}
}

\begin{document}
\maketitle

\begin{abstract}
	Many cell functions are accomplished thanks to intracellular transport mechanisms of macromolecules along filaments.
	Molecular motors such as dynein or kinesin are proteins playing a primary role in these processes.
	The behavior of such proteins is quite well understood when there is only one of them moving a cargo particle. Indeed, numerous in vitro experiments have been performed to derive accurate models for a single molecular motor.
	However, in vivo macromolecules are often carried by multiple motors.
	The main focus of this paper is to provide an analysis of the behavior of more molecular motors interacting together in order to improve the understanding of their actual physiological behavior.
	Previous studies provide analyses based on results obtained from Monte Carlo simulations. Different from these studies, we derive an equipollent probabilistic model to describe the dynamics of multiple proteins coupled together and provide an exact theoretical analysis. We are capable of obtaining the probability density function of the motor protein configurations, thus enabling a deeper understanding of their behavior.
\end{abstract}

\section{Introduction}\label{sec:intro}
Many biological cell functions are accomplished thanks to intracellular transport processes. A key role in these processes is played by molecular motors that are proteins converting chemical energy into mechanical work.
Among the most important molecular motors, there are kinesin and dynein.
These are classes of molecular motors transporting many kind of different molecules by moving along filaments called microtubuli.
They accomplish this by converting the chemical energy stored in the ATP molecules produce mechanical work and thus represent powerful nanomachines.
Both kinesin and dynein have a similar structure consisting of two catalytically active motor domains, called ``heads'' that can bind both to a microtubule and ATP. The two heads are hinged together by a polypeptidic linkage that terminates in a distal cargo binding domain. These motor heads bind and unbind to the microtubule successively by utilizing the energy available from ATP hydrolysis and enable the motor to walk \cite{CarCro2005}.

The study of molecular motors has developed over the past decade and these investigations have unraveled many previously unknown finer details at the microscopic level which have propelled various theoretical discussions of the mechanisms underlying the dynamics of molecular motors. Given the small sizes of molecular motors (each head of kinesin has a linear size of the order of 10 nm), they are supposed to be perturbed by the state of the local molecular environment and thermal fluctuations around them. But the directed walk of these cytoskeletal motors clearly reveals that they are robust enough to overcome the fluctuations surrounding them.
The behavior of motor proteins is quite well understood when there is only one of them moving a cargo particle. Numerous in vitro experiments have been performed to derive accurate models to describe the motion of a single molecular motor. Indeed, it is possible to monitor the motion of a single molecular motor under highly tunable experimental conditions and obtain measurements with accurate spatial and time resolution \cite{SvoSch93, MilYil06a, CarVer08}.
As a result, there are many theoretical descriptions of motor-protein mechanisms that recognize their multiple conformational transitions taking into account the complex mechanochemical processes involved \cite{LieLip09}.\\
However, the chemical and mechanical properties of molecular motors in vivo are not yet completely characterized. Furthermore, there is evidence that, in vivo, macromolecules are carried by multiple motors \cite{KunVer08}.
Indeed, this can be a cause of the robustness of their behavior and of also a cause of increased efficiency.
The possible dynamics of an ensemble of motor proteins carrying the same cargo particle is still largely unknown. For example, it is not clear whether synchronization phenomena occur, the motors move independently or a ``tug-of-war'' is in action \cite{KurKim05, SopRai09}.
Understanding these mechanisms is one of the more challenging problems that require concerted efforts by chemists, physicists, biologists and engineers.

\section{Markov model for molecular motors}
In this section we derive a Markov model to describe the motion of an ensemble of molecular motors carrying a cargo.
While the behavior of single molecular motors such as kinesin or dynein is quite well understood, we intend to focus on the possible properties of synchronization/coordination that can be exhibited by more motors connected to the same cargo.\\
Molecular motors are proteins moving macro-molecules and organelles along filaments (i.e. microtubuli).
Their motion occurs by discrete steps. Their heads move forward by hydrolyzing ATP and binding to proper locations that are equally spaced on the microtubule. Every motor can attach to the cargo molecule via a flexible linkage (see Figure~\ref{fig:SteppingMotor}).
\begin{figure}[hbt]
	\includegraphics[width=0.95\columnwidth]{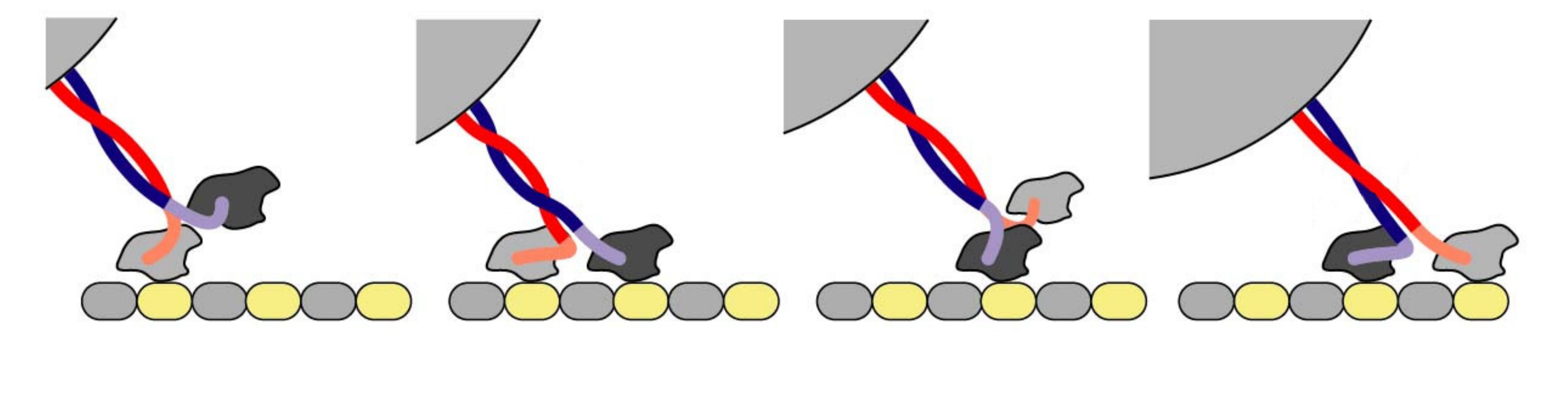}
	\caption{
		A molecular motor making a step
		\label{fig:SteppingMotor}
	}
\end{figure}

In \cite{KunVer08} a stochastic algorithm simulating the dynamics of multiple interacting motors was employed to obtain useful insights about their behavior via Monte Carlo simulations.
We ``borrow'' the rules of the algorithm given by \cite{KunVer08} and introduce a rigorous mathematical formalism in order to derive a Markov model and investigate how multiple motors function together.
However, in contrast to \cite{KunVer08}, our approach will have the advantage of relying on a probabilistic analysis that is exact and is not based on Monte Carlo simulations.\\
Every motor is modeled as a particle moving on a line by discrete steps. Its position can assume values that are equally spaced by the unit step size $d_{s}$.
Multiple motors are allowed to bind to the same location.
Every motor linkage behaves as a spring with elastic constant $K$ and rest length $l_0$ with no compressional rigidity: when compressed it buckles down with no resistance.
It is assumed that the motors are bound to the same cargo particle. The cargo is subject to the forces applied by the motors and to a constant load $F_{load}$ that is assumed positive when it opposes the motor movement.\\
In \cite{KunVer08}, two different models are provided for the cargo motion. In a simplified model the cargo moves deterministically reaching its equilibrium position immediately. In a more realistic model the cargo is subject to a Brownian force that models the thermal fluctuations of the environment.\\
In this paper we focus on the first scenario only, even though our results can be extended to treat the noisy case, as well.
The motors move by single steps of length $d_{s}$ in a stochastic manner. 
Backward steps are not allowed. 
Given two non-overlapping time intervals, the probabilities that a motor takes a step in each of them are independent.
This leads to a Poisson distribution of the time after which every next step occurs.
For each motor, the stepping rate of the Poisson distribution depends on the reaction force $F$ it is subject to from the cargo. $F$ is assumed positive when it opposes the motor movement.
A motor subject to a force $F$ will have a stepping rate given by $\lambda^{(step)}(F)$.
Estimates for $\lambda^{(step)}(F)$ can be obtained experimentally by single motor experiments and from the knowledge of the chemical process behind the ATP hydrolyzation.\\
Given the forces acting on each motor, they step independently from each other.
The motors have a ``stalling force'' $F_s$ such that, for $F\geq F_s$, we have $\lambda_{step}(F)=0$. The stalling force can be determined via experiments, as well.
The force applied to a motor also influence the probability rate $\lambda^{det}(F)$ at which it can detach form the microtubule. Every unengaged motor can also reattach to the microtubule at a rate $\lambda^{att}$, but only at locations that are distant not more than $l_0$ from the cargo position. In other words, a motor can not attach to the microtubule at a location where its linkage is stretched.
Finally, the model assumes that at any time there can not be more than $\overline m$ motors engaged on the microtubule.

\subsection{Mathematical formalism and Notation}
We model the microtubule locations as a bi-infinite sequence $\{a_k\}_{k\in \Z}$, $a_k\in\R$, where each $a_k$ represents the position of a binding site at which a motor can be located.
Since the locations are equally spaced by the distance $d_{s}$, we provide the following definition
\begin{define}
	We define a {\it location sequence} as $A:=\{a_k\}_{k\in \Z}$, where
	\begin{align}
		a_k=x^{(0)}+kd_{s}
	\end{align}
	and $x^{(0)}\in\R$.
\end{define}
Note that $x^{(0)}\in\R$ is determined by the arbitrary choice of the reference point on the microtubule.\\
The arrangement of the motors on the microtubule is described by a bi-infinite sequence of natural numbers $\{z_k\}_{k\in\Z}$, $z_k\in\N$, where $z_k$ represents the number of motors bound at the position $a_k$.
Equivalently, the arrangement of motors on the microtubule can be described by a finite ordered sequence $\{x_i\}_{i=1}^{m}$ representing the position of each of the $m$ motors engaged on the microtubule.
\begin{define}
	Assume a location sequence $A=\{a_k\}_{k\in \Z}$.
	We refer to a {\it configuration} as any bi-infinite sequence $Z:=\{z_k\}_{k\in\Z}$ of natural numbers such that
	\begin{align}
		m:=\sum_{k\in\Z}z_k<\infty.
	\end{align}
	The natural number $m$ is the {\it number of engaged agents} of the configuration $Z$.\\
	We define the {\it position sequence} of the configuration $Z$ with respect to the location sequence $A$ as the unique non-decreasing finite sequence $X:=\{x_i\}_{i=1}^{m}$ such that
	\begin{align}
		z_k=|\{i\in \N \text{ such that } x_i\in X \text{ and } x_i = a_k \}|,
	\end{align}
	where $|\cdot |$ denotes the cardinality of a set. We write $X=\phi_{A}(Z)$.
\end{define}
Each representation (configuration $Z$ or position sequence $X$) has notation advantages and drawbacks. For example, given a configuration, transitions like steps, attachments and detachments will be described by adding a proper ``transition sequence''. Conversely, to determine the forces the cargo particle is subject to the representation it terms of a position sequence will be easier to handle.
\begin{define}
	We define $\mathcal{Z}_{m}$ as the set of all the possible configurations with $m$ engaged agents and $\mathcal{X}_{m}$ as the set of all possible position sequences with $m$ engaged agents. We also define
	\begin{align}
		& \mathcal{Z}:= \bigcup_{m=1}^{+\infty} \mathcal{Z}_{m}
		\qquad
		& \mathcal{X}:= \bigcup_{m=1}^{+\infty} \mathcal{X}_{m}.
	\end{align}
\end{define}
\begin{lem}
	The set $\mathcal{Z}$ is countable.
\end{lem}
\begin{proof}
	The proof is left to the reader.
\end{proof}
\begin{lem}
	For any $Z\in\mathcal{Z}$, let $X=\phi_{A}(Z)\in\mathcal{X}$ be the position sequence of the configuration $Z$ with respect to the location sequence $A$.
	The function $\phi_{A}:\mathcal{Z}\rightarrow \mathcal{X}$ is a one-to-one correspondence. Moreover, $\phi_{A}(\mathcal{Z}_m)=\mathcal{X}_m$ for any $m\in\N$.
\end{lem}
\begin{proof}
	The proof is left to the reader.
\end{proof}

\begin{define}
	For any $\alpha\in\Z$ and $Z=\{z_k\}_{k\in\Z}\in\mathcal{Z}$, we define the shift operator $\rho^{\alpha}:\mathcal{Z}\rightarrow\mathcal{Z}$
	\begin{align}
		\rho^{\alpha}(Z):=\{z_{k+\alpha}\}_{k\in\Z}.
	\end{align}
\end{define}

We assume that the motors move towards larger values of their positions.
With this assumption, we introduce the following definition.
\begin{define}
	Consider $Z=\{z_k\}_{k\in\Z}\in\mathcal{Z}$ and let $X=\phi_A(Z)=(x_1,...,x_m)$ be its position sequence with respect to the location $A$.
	We say that $x_m$ is the {\it vanguard agent} position while $x_1$ is the {\it rearguard agent} position of the configuration $Z$ (and of the position sequence $X$).
	We define {\it spread} of $Z$ (or spread of $X$ abusing the nomenclature) as
	\begin{align}
		s(Z):=x_{m}-x_1
	\end{align}
	and the dimension of $Z$  as
	\begin{align}
		h(Z):=1+\sup\{k\in\Z| z_k\neq 0\}-\inf\{k\in\Z|z_k\neq 0\}.
	\end{align}
	If $m=0$ we define $s(Z):=0$ and $h(Z):=1$.
\end{define}

In the model, only certain transitions are allowed, representing  a forward step, a motor detachment from the microtubule and a motor attachment to the microtubule.
\begin{define}
	Define the bi-infinite sequences
	\begin{align}
		& R^{(s)}_k:=\delta_{k+1}-\delta_{k}\\
		& R^{(d)}_k:=-\delta_{k}\\
		& R^{(a)}_k:=\delta_{k}
	\end{align}
	where $\delta_{k}$ represents the bi-infinite Kronecker delta sequence that is equal to $1$ at $k$ and zero anywhere else.
\end{define}

\begin{define}
	Consider a time-continuous Markov Process defined on the state space $\mathcal{Z}$, with time variable $t\in\R$
	and with transition rate from the configuration $Z_1$ to $Z_2$ given by $\lambda(Z_2,Z_1)$, meaning that, given that the system is in the configuration $Z_1$ at time, the probability of being in $Z_2$ at time $t+\Delta t$ is equal to $\lambda(Z_2,Z_1)\Delta t$ (for small $\Delta t$).\\
	Assume that, for any $Z_2, Z_1$ such that $\lambda(Z_2,Z_1)>0$,
	there exists $k\in\Z$ such that one of the following conditions is met
	\begin{itemize}
		\item $Z_2 = Z_1+R^{(s)}_k$
		\item $Z_2 = Z_1+R^{(d)}_k$
		\item $Z_2 = Z_1+R^{(a)}_k$.
	\end{itemize}
	We call the transitions sequences $R^{(s)}_k$, $R^{(d)}_k$ and $R^{(a)}_k$ respectively {\it forward step} from location $k$, {\it detachment} from location $k$ and {\it attachment} to location $k$.
	We call such a process an {\it Ensemble of Molecular Motors} (EMM).
\end{define}
The cargo transported by the group of molecular motors is usually subject to a constant force $F_{load}$. In this simplified article it is assumed that, after any transition, the cargo reaches its equilibrium position instantaneously. We provide the following definition.
\begin{define}
	Assume that a motor in position $x$ is subject to a force equal to $\chi(x-x^{(c)})$ given that the cargo is in position $x^{(c)}\in\R$.\\
	Given a position sequence $X=\{x_1,...,x_m\}$ (or equivalently its configuration $Z=\phi_A^{-1}(X)$), an equilibrium position of the cargo $x^{(c)}$ is any real number satisfying the equilibrium condition
	\begin{align}\label{eq:equilibrium}
		F_{load}=\sum_{i=1}^{m} \chi(x_i-x^{(c)}).
	\end{align}
\end{define}

\begin{define}
	We refer to any realization $Z(t)$ of a EMM as a configuration path.
	In an obvious way, any configuration path $Z(t)$ can be associate to a {\it position path} $X(t)=\phi_A(Z(t))=(x_1(t),...,x_{m(t)}(t))$. For any $t$, we also have the number of engaged agents $m(t)$, the vanguard agent with position $x_{m(t)}(t)$, the rearguard agent $x_{1}(t)$ and, under the assumption that it is unique, the cargo equilibrium position $x^{(c)}(t)$.
\end{define}

In \cite{KunVer08} a set of specific rules are defined to create a stochastic simulation algorithm for the motion of molecular motors. The rules are the following.
The cargo has $\overline{m}$ attached motors that can attach and detach from the microtubule. Their linkages behave as elastic springs of elastic constant $K$ and rest length $l_0$ that buckle down with no resistance when compressed.
It is postulated that motors engaged on the microtubule have always a chance of detaching depending on the force they are subject to. On the other hand, any disengaged motor can attach again, but only at locations where its linkage is not subject to a stretch (that is at distance not larger than $l_0$ from the cargo). The probability of taking a step depends on the force acting on the specific motor and there exists a ``stalling force'' $F_s>0$ making the step not possible. Finally, the behavior of the system is invariant to translations along the microtubule.
These rules are formalized in the following way
\begin{define}
	Given a EMM, we say that it satisfies the Kunwar-Veshinin-Xu-Gross (KVXG) rules with parameters $(\overline{m}, K, l_0, F_s)$, with $\overline{m}\in\N$ and $K, l_0, F_s$ real positive, if it meets the following properties.
	Assume a location sequence $A=\{a_k\}_{k\in Z}$; for any $Z_1\neq Z_2 \in \mathcal{Z}$, with position sequences $X_1$, $X_2$ and engaged agents $m_1$, $m_2$ respectively, we have
	\begin{itemize}
		\item	irreversible cargo loss with no engaged motors
		\begin{align}
			m_1=0 \text{ implies } \lambda(Z_2,Z_1)=0
		\end{align}
		\item	at most $\overline{m}$ engaged motors
		\begin{align}
			m_2>\overline{m} \text{ implies } \lambda(Z_2,Z_1)=0
		\end{align}
		\item	spatial invariance (for any $\alpha\in\Z$)
		\begin{align}
			\lambda(Z_2,Z_1)=\lambda(\rho^{\alpha}Z_2,\rho^{\alpha}Z_1)
		\end{align}
		\item	elastic stretch with constant $K$ and dead zone $l_0$
		\begin{align}\label{eq:elastic dead zone}
			\chi(l):=
			\left\{\begin{array}{ll}
				K(l-l_0)	& \text{if } l>l_0\\
				0		& \text{if } |l|\leq l_0\\
				K(l+l_0)	& \text{if } l<-l_0.
			\end{array}\right.
		\end{align}
		\item	attachment with no stretch
		\begin{align}
			& \lambda(Z_2,Z_1)>0 \text{ and } Z_2-Z_1=R^{(a)}_k \nonumber\\
			& \text{ implies } |a_k-x^{(c)}|\leq l_0
		\end{align}
		\item	stalling force $F_s$
		\begin{align}
			\chi(a_k-x^{(c)})\geq F_s \text{ implies } \lambda(Z_1+R^{(s)}_k,Z_1)=0.
		\end{align}
	\end{itemize}
\end{define}

Given an EMM with a transition rate function $\lambda:\mathcal{Z}\rightarrow \R^{+}$ meeting the KVXG rules, define $P_Z(Z,t|\overline{Z}, t_0)$ as the probability of being in the configuration $Z=\{z_k\}_{k\in\Z}$ at time $t\geq t_0$ conditioned to $Z(t_0)=\overline{Z}$.

Then, given an initial time $t_0$ and an initial state $\overline{Z}$, for $t\geq t_0$, $P_{Z}(Z,t|\overline{Z},t_0)$ satisfies the Master Equation
\begin{align}\label{eq:infinite master equation}
	& \frac{\partial}{\partial t}P_{Z}(Z,t|\overline{Z},t_0)=
		-P_{Z}(Z,t|\overline{Z},t_0)
			\sum_{Z'\in\mathcal{Z}}\lambda(Z',Z)\\
		&\qquad +\sum_{Z'\in\mathcal{Z}}\lambda(Z,Z')
			P_{Z}(Z',t|\overline{Z},t_0),
\end{align}
that represents the conservation law of the probability measure.

\section{Finite dimensional projection and Master Equation analysis}
The study of the dynamics of the motors (represented either in the configuration space or in the position sequence space) is not easy to handle since the entities involved are not stationary and the state-space is infinite-dimensional.
The following theoretical result is helpful for transforming the problem into a finite dimensional one.
\begin{theo}\label{thm:bounded configuration}
	Consider an EMM satisfying the KVXG rules with parameters $(\overline{m}, K, l_0, F_s)$. For any configuration path $Z(t)$ let $s(Z(t))$ be its associated spread for any $t$.
	Define
	\begin{align}
		s^{(max)}:=\max\left\{\frac{\overline{m} F_s-F_{load}}{K}+d_{s},\frac{F_{load}}{K}\right\}+2l_0.
	\end{align}
	For any $S\geq s^{(max)}$ and for any $t\geq t_0$, we have that if $s(\overline{Z})\leq S$
	\begin{align}
		Pr\{s(Z(t))\leq S~|~Z(t_0)=\overline{Z}\}=1.
	\end{align}
	where $Pr\{E|C\}$ is the conditional probability of an event $E$ given the condition $C$.
\end{theo}
\begin{proof}
	See the appendix
\end{proof}
Theorem \ref{thm:bounded configuration} enables us to study many properties of an EMM using a formulation that has finite dimension and deals with stationary quantities.
To this aim, we introduce the following definitions.
\begin{define}
	Given an EMM satisfying the KVXG rules with parameters $(\overline{m}, K, l_0, F_s)$ subject to load $F_{load}$ and a location sequence $A$ with step size $d_s$, we define
	\begin{align}
		n:=\left\lceil
			\frac{s^{(max)}}{d_{s}}
		\right\rceil.
	\end{align}
	as the {\it regular dimension} of the ensemble.\\
	Given a configuration $Z=\{z_k\}_{k\in\Z}$, we also define its {\it ensemble representation} as the $n$-vector
	\begin{align}
		Q=(z_{k_1}, ... z_{k_1+n-1})^T
	\end{align}
	where $k_1:=\inf\{k\in\Z | z_k\neq 0\}$ and we write $Q:=\Pi^{(e)}(Z)$.
	If $k_1=-\infty$, $Q$ is $(0, 0, ..., 0)^T$.
\end{define}
The ensemble representation of $Z$ is the projection of the configuration on a suitable $n$-dimensional space.
Theorem \ref{thm:bounded configuration} guarantees that the ensemble representation $Q(t)$ of the configuration $Z(t)$ contains all the non-zero entries of $Z(t)$ with probability $1$ for all $t>t_0$ if the dimension of $Z(t_0)$ is not greater than $n$.
In this scenario, while $Q(t)$ does not contain the information about the absolute positions of the motors, it still contains all the information about all their relative positions, so it can be used as a tool to study the dynamics of their mutual interactions.\\
Assume that the starting configuration $Z(t_0)=\overline{Z}$ has dimension not exceeding the regular one $n$.
For $t>t_0$ the spread is not exceeding $n$ and the transition probabilities of $Z(t)$ are known. Despite the projection, it is possible to derive the transition probabilities of $Q(t)$. This allows to define the dynamics of a Markov system with a finite dimension.
The following properties hold.
\begin{lem}
	Consider an EMM with regular dimension $n$. Let $s(\overline{Z})\leq n$ and
	$P_{Q}(Q,t)$ be the probability that $\Pi^{(e)}(Z)$ is equal to $Q$ at time $t$.
	Then, for $Q\neq 0$ and for any $t\geq t_0$
	\begin{align}
		P_Q(Q,t|\overline{Z},t_0)=
		\sum_{\alpha\in\Z}P_Z\{\rho^{\alpha}Z,t|\overline{Z},t_0\}
	\end{align}
\end{lem}
\begin{proof}
	The proof is left to the reader.
\end{proof}

\begin{theo}
	Consider an EMM with regular dimension $n$ and let $s(t)$ be the spread at time $t$.
	Define 
	\begin{align}
		\lambda_{Q}(Q',Q):=\sum_{\stackrel{\Pi^{(e)}(Z')=Q'}{s(Z)\leq n}}\lambda(Z',Z).
	\end{align}
	For any $t\geq t_0$ and for any $\overline{Z}$ such that $s(\overline{Z})\leq n$, it holds that
	\begin{align}
		&\frac{\partial}{\partial t}P_{Q}(Q,t|\overline{Z},t_0)=
			-P_{Q}(Q,t|\overline{Z},t_0)\sum_{Q'\in\mathcal{Q}}
				\lambda(Q',Q)\\
		&\qquad+\sum_{Q'\in\mathcal{Q}}
			\lambda(Q,Q')P_Q(Q',t|\overline{Z},t_0)
	\end{align}
\end{theo}
\begin{proof}
	In order to simplify the notation we omit the initial condition: $P_Q(Q,t)=P_Q(Q,t|\overline{Z},t_0)$ and $P_Z(Z,t)=P_Z(Z,t|\overline{Z},t_0)$.
	We have that $Q=0$ if and only if $Z=0$. In such a case, $\lambda(Z',Z)=0$ for any $Z'$ and $\lambda(Z,Z')=0$ for any $Z'$ such that $\Pi^{(e)}(Z')\neq Q_{1}:=(1, 0, ... ,0)$.
	Observe that
	\begin{itemize}
		\item $\lambda_{Q}(Q',0)=0$ for any $Q'$
		\item $\lambda_{Q}(0,Q')=0$ for any $Q'\neq Q_{1}$
		\item $P_Z(Z',t)=0$ if $s(Z')>n$.
	\end{itemize}
	Then, for $Q=0$ we obtain 
	\begin{align*}
		&\frac{\partial}{\partial t}P_{Q}(0,t)=
			\frac{\partial}{\partial t}P_{Z}(0,t)=\\
		&=\sum_{\stackrel{\Pi^{(e)}(Z')=Q_1}{s(Z)\leq n}}
			\lambda(0,Z')P_{Z}(Z',t)=\lambda_{Q}(0,Q_1)P_{Q}(0,Q_1).
	\end{align*}
	For $Q\neq 0$, we have
	\begin{align*}
		&\frac{\partial}{\partial t}P_{Q}(Q,t)=
			\frac{\partial}{\partial t}P_{Q}(\Pi^{(e)}(Z),t)
		=\sum_{\alpha\in\Z}\frac{\partial}{\partial t}P_{Z}(\rho^{\alpha}Z,t)=\\
		&\qquad -\sum_{\alpha\in\Z}
			\sum_{Z'\in\mathcal{Z}}
				\lambda(Z',\rho^{\alpha}Z)P_{Z}(\rho^{\alpha}Z,t)\\
		&\qquad +\sum_{\alpha\in\Z}
			\sum_{Z'\in\mathcal{Z}}
				\lambda(\rho^{\alpha}Z,Z')P_{Z}(Z',t)=\\
		&=-\sum_{\alpha\in\Z}
			\sum_{Z'\in\mathcal{Z}}
				\lambda(\rho^{\alpha}Z',\rho^{\alpha}Z)P_{Z}(\rho^{\alpha}Z,t)\\
		&\qquad+\sum_{\alpha\in\Z}
			\sum_{Z'\in\mathcal{Z}}
				\lambda(\rho^{\alpha}Z,\rho^{\alpha}Z')P_{Z}
				(\rho^{\alpha}Z',t)\\
		&=-\sum_{Z'\in\mathcal{Z}}
			\lambda(Z',Z)
				\sum_{\alpha\in\Z}P_{Z}(\rho^{\alpha}Z,t)\\
		&\qquad+\sum_{Z'\in\mathcal{Z}}
			\lambda(Z,Z')
				\sum_{\alpha\in\Z}P_{Z}(\rho^{\alpha}Z',t)=\\
		&=-P_{Q}(Q,t)
			\sum_{Z'\in\mathcal{Z}}
				\lambda(Z',Z)\\
		&\qquad+\sum_{\stackrel{Z'\in\mathcal{Z}}{s(Z')\leq n}}
			\lambda(Z,Z')P_Q(\Pi^{(e)}(Z'))=\\
		&=-P_{Q}(Q,t)\sum_{Q'\in\mathcal{Q}}
				\lambda_{Q}(Q',Q)\\
		&\qquad+\sum_{Q'\in\mathcal{Q}}
			\lambda_{Q}(Q,Q')P_Q(Q')
	\end{align*}
\end{proof}

\begin{lem}
	Given an EMM satisfying the KVXG rules with parameters $(\overline{m}, K, l_0, F_s)$ with regular dimension $n$,  the number of possible ensemble representations is given by
	\begin{align}\label{eq:ensemble configuration number}
		N=1+\sum_{m=1}^{\overline m}\frac{(n+m-2)!}{(n-1)! (m-1)!}.
	\end{align}
\end{lem}
\begin{proof}
	The proof is left to the reader.
\end{proof}

Enumerate all the possible ensemble representations $Q_1,...,Q_N$ with $N$ given by (\ref{eq:ensemble configuration number}) and define $P_{i}(t)$ as the probability of having the ensemble representation $Q_i$ at time $t$ given the initial condition $\overline{Z}$ such that $s(\overline{Z})\leq n$.
By projecting Equation (\ref{eq:infinite master equation}) onto the set $\mathcal{Q}:=\{Q_1,...,Q_N\}$ we obtain 
\begin{align}\label{eq:finite master equation}
	\frac{d}{dt}P_{Q}(Q,t)=
		-\sum_{Q'\in\mathcal{Q}^{(+)}(Q)}\lambda_{Q}(Q',Q)P_{Q}(Q,t)\\
		+\sum_{Q'\in\mathcal{Q}^{(-)}(Q)}\lambda_{Q}(Q,Q')P_{Q}(Q',t),
\end{align}
The dynamics of the probabilities on the space of the ensemble representations is given by an equation of the form
\begin{align}\label{eq:master equation}
	\frac{d}{dt}P=A P
\end{align}
where $P:=(P_1,...,P_N)^T$, the off-diagonal element of $A$, $a_{ji}$ is the rate of the transition from $Q_i$ to $Q_j$ and the sum of all the rows of $A$ is zero.
In this way the dynamics of the infinite dimensional Markov system has been projected onto the dynamics of a finite dimensional one. In the section we show how this projection provides no loss of information when used to obtain certain quantities of interest such as the expected velocity of the cargo and the expected runlength.

\section{Determining the transition probabilities from single motor observations}
In this section we determine the transition rates $\lambda(Z',Z)$ for the infinite dimensional stochastic model.
The transition rates $\lambda_{Q}(Q',Q)$ necessary to define the finite dimensional one can be computed from the transition rates $\lambda(Z',Z)$.
Following \cite{KunVer08}, it is possible to derive the probability rates of simple events $\lambda(Z',Z)$ by considering the behavior of a single motor.
\subsection*{Rate for stepping transitions}
Every motor moves converting ATP into kinetic energy according to the chemical equation
\begin{equation}
	K+ATP \stackrel[k_{off}]{k_{on}}{\rightleftarrows}
	K~ATP \stackrel{k_{cat}}{\rightarrow}
	K+ADP+P_{i}+{\tt energy}.
\end{equation}
Following \cite{MeyHow95}, a Michaelis-Menten dynamics leads to a hydrolisis rate equal to
$k_{cat}[ATP] /([ATP]+k_m)$.
It is also assumed that the free head of the motor successfully binds to the microtubule location with probability rate (or efficiency) $\eps$.
Then, the transition rate $\lambda_{step}$ of stepping and the average velocity $V$ of a single motor are given by
\begin{align}
	& \lambda_{step}=\frac{k_{cat}[ATP]}{[ATP]+k_m}\eps\\
	& V=P_{step}d=\frac{k_{cat}[ATP]}{[ATP]+k_m}d\eps
\end{align}
where $d$ is the fixed size of the protein step.
The force $F$ that the cargo exerts on the motor is assumed positive when it opposes the motor motion. It is assumed that when it exceeds a certain force $F_{0}$ it causes the motor to stall.
Following \cite{KunVer08}, it is assumed to affect the motor dynamics by changing the probability $\eps$ of binding with the microtubule, following the relation
\begin{align}
	\eps(F)=
		\left\{\begin{array}{lr}
			1			&\text{if } F\leq 0\\
			1-\left(\frac{F}{F_{s}}\right)^2  & \text{if }0<F<F_{0}\\
			0	& \text{otherwise}
		       \end{array}
		\right.
\end{align}
where $F_{s}$ is the stalling force of the KVXG rules. 
In Figure~\ref{fig:efficiency}, it is represented how the efficiency $\eps(F)$ changes in different situations: when the motor is lagging behind the cargo, when it is at a distance where the linkage is not stretched and when it is actively pulling the cargo.\\
\begin{figure}
	\centering
	\includegraphics[width=0.6\columnwidth]{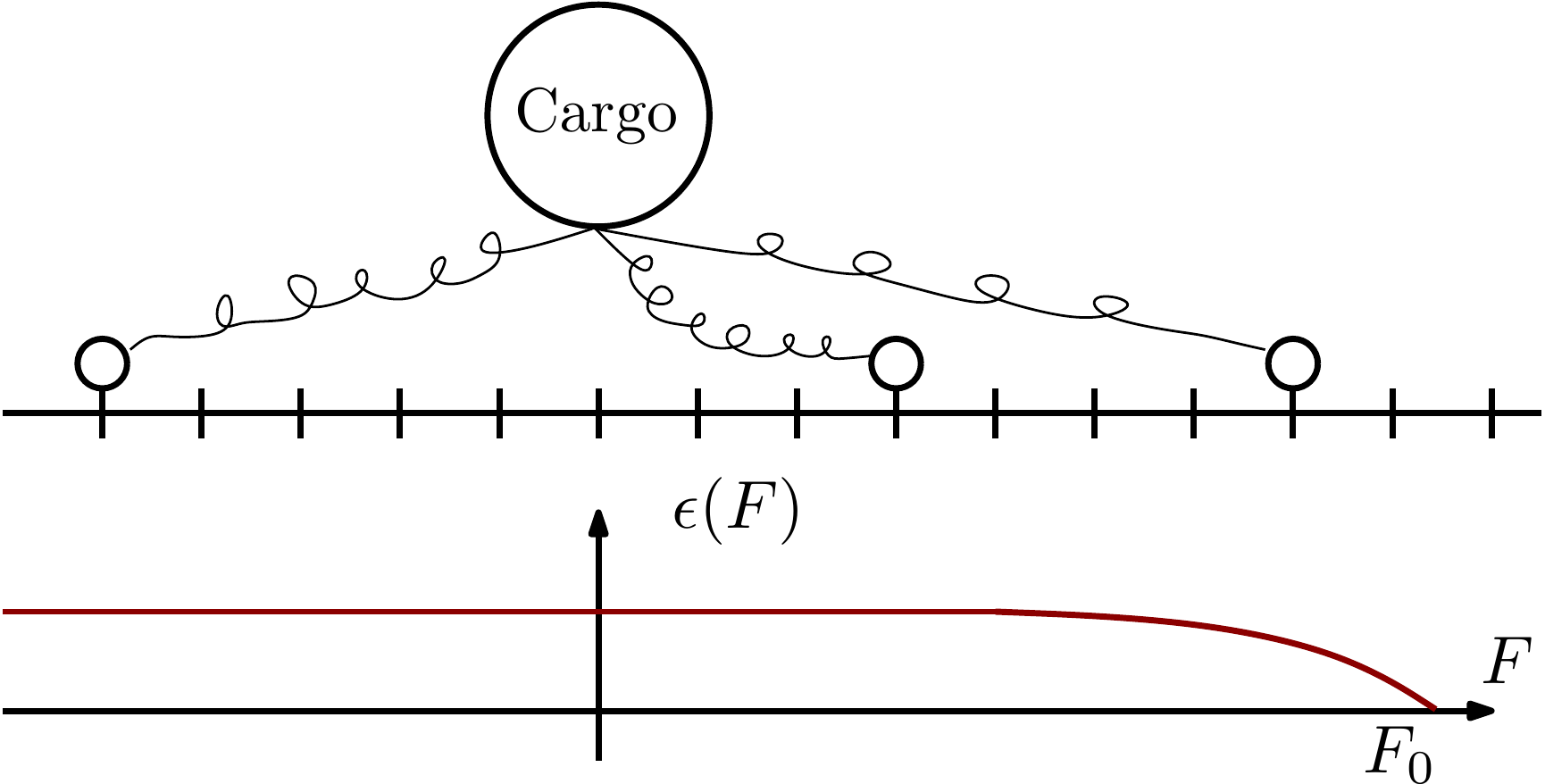}
	\caption{Graphical representation of the ``efficiency'' $\eps(F)$ and its dependence on the force $F$ applied to the motor
	\label{fig:efficiency}
	}
\end{figure}
The linkage is modeled as an elastic spring of elastic constant $K$ and rest length $l_0$ that buckles down with no resistance when compressed. The relation between the force applied to the motor and the linkage length $l$ is given by
\begin{align}
	F(l)=	\left\{\begin{array}{lr}
			K(l+l_0)	&\text{if } l\leq -l_0\\
			0		&\text{if } |l|<l_{0}\\
			K(l-l_0)	&\text{if } l\geq l_0.\\
		\end{array}\right.
\end{align}
\begin{figure}
	\centering
	\includegraphics[width=0.5\columnwidth]{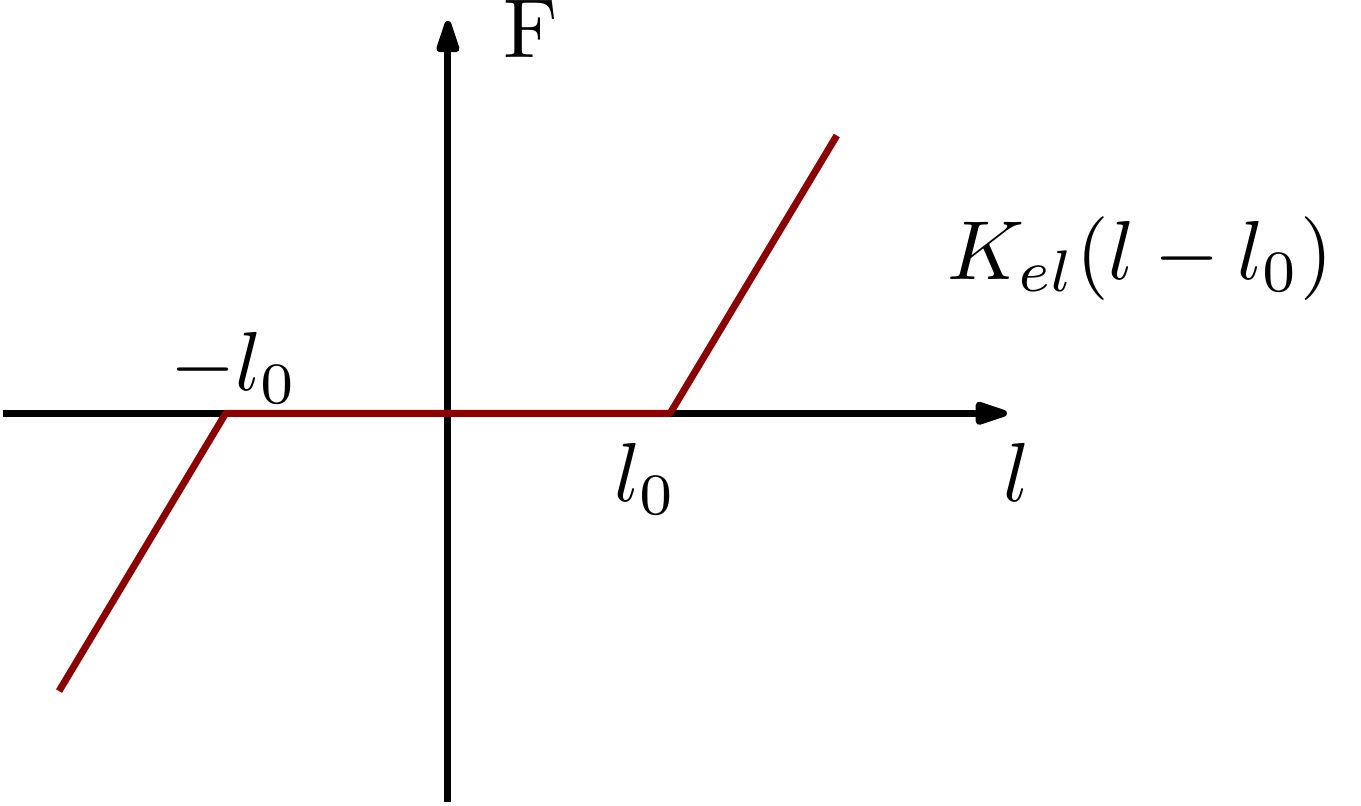}
\end{figure}
In \cite{KunVer08} it is assumed that the force $F$ also influences the kinetics of the ATP hydrolysis. In particular it is assumed that $k_{off}$ increases with increasing $F$
\begin{equation}
	k_{off}=k_{0off}e^{Fd_l/K_bT},
\end{equation}
where $k_{0off}$ is the backward reaction rate of the hydrolysis when $F=0$, $K_b$ is the Boltzmann constant, $T$ is the temperature and $d_l$ is a parameter that can be experimentally determined.
Thus, the transition rate for a step at location $a_k $ is given by
\begin{align}\label{eq: lambda step}
	&\lambda(Z+R_{k}^{(s)},Z)=
		z_k\frac{k_{cat}[ATP]}
		{[ATP]+\frac{k_{on}+k_{off}(F)}{k_{cat}} }\epsilon(F)\\
	&=
		z_k\frac{k_{cat}[ATP]\left[1-\left(\frac{\chi(a_k-x^{(c)})}{F_0}\right)^2\right]}
		{[ATP]+\frac{k_{on}+k_{0off}e^{\chi(a_k-x^{(c)})d_l/K_bT}}{k_{cat}} }
\end{align}

\subsubsection*{Rate for detachment transitions}
From [Schnitzer et al. 2000], the processivity $L$ is
\begin{align}
	L=\frac{d[ATP]Ae^{-F \delta_l K_bT}}
		{[ATP]+B(1+A)e^{-F \delta_l K_bT}},
\end{align}
where $A$, $B$ and $\delta_{l}$ are parameters that can be experimentally determined.
Since the processivity represents how far a motor can move before detaching from the microtubule on average, we find a relation between the probability of stepping and the probability of detachment.
\begin{align}
	\frac{P_{step}(F)}{P_{detach}(F)}=\frac{L}{d}=
		\frac{[ATP]Ae^{-F \delta_l K_bT}}
		{[ATP]+B(1+A)e^{-F \delta_l K_bT}}.
\end{align}
Thus, so long as $F<F_0$, 
\begin{align}
	P_{detach}(F)=\frac{[ATP]+B(1+A)e^{-F \delta_l K_bT}}{[ATP]Ae^{-F \delta_l K_bT}}P_{step}(F).
\end{align}
When $F\geq F_s$, in \cite{KunVer08} a constant detachment rate is assumed $P_{detach}(F)=P_{back}=2s^{-1}$.
Thus, we have the following transition rates
\begin{align}
	&\lambda(Z+R_{k}^{(d)},Z)=
		\left\{\begin{array}{l}
			\frac{[ATP]+B(1+A)e^{-\chi(a_k-x^{(c)}) \delta_l K_bT}}{[ATP]Ae^{-\chi(a_k-x^{(c)} \delta_l K_bT}} \lambda(Z+R_{k}^{(s)},Z)\\
					\qquad \text{ if } \chi(a_k-x^{(c)} < F_s\\
			z_{k} P_{back} \\
					\qquad \text{ if } \chi(a_k-x^{(c)} \geq F_s
		\end{array}\right.
\end{align}

\subsubsection*{Probability of attachment}
Experimentally, it is found in \cite{KunVer08} that the probability of a motor to attach to the microtubule is $P_{att}\simeq 5 s^{-1}$. If the motor is linked to the cargo, it is assumed that is attaches to the microtubule without stretching its linkage. Thus, the only admissible locations of attachment are the locations at a distance from the cargo that is less than $l_0$. They are also assumed all equally likely.

\section{Probabilistic Analysis}\label{sec:analysis}
We have shown that the length of all the possible strings can be uniformly bounded.
As a consequence, apart from possible spatial shifts along the microtubule, the number of possible arrangements of the motors is finite.

Using the finite dimensional Markov Model derived in the previous sections, it is possible to derive meaningful characteristics of the system in an exact manner.

\subsection{Transient analysis}
In Figure~\ref{fig: time histogram 3 motors} we present the dynamics of $P(t)$  in the case $\overline{m}=3$ for two different values of $F_{load}$.
We use $4$ time snapshots to describe how $P(t)$ changes with time.
\begin{figure*}
	\centering
	\includegraphics[width=0.95\textwidth]{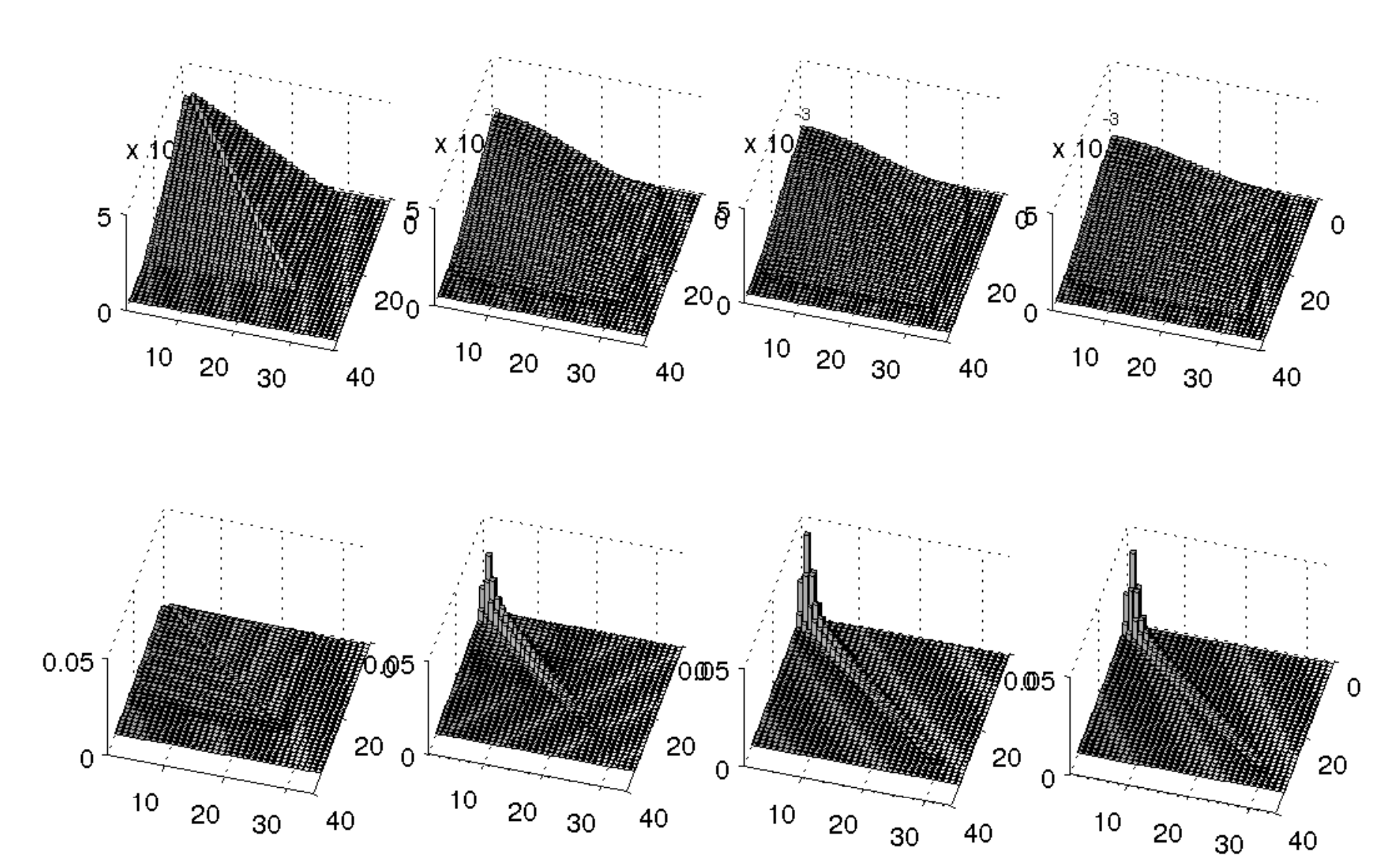}
	\caption{Four snapshots of the evolution of the probability density function in the case of low load (top row) and in the case of high load (bottom row)
	\label{fig: time histogram 3 motors}
	}
\end{figure*}
In the first row of Figure~\ref{fig: time histogram 3 motors} we present the probability density function of the configurations 
when the cargo is subject to a low load ($F_{load}=0.0002 nN$). The $x$ and the $y$ axis represent the distance of the two leading motors from the rearguard motor. 
The distance is expressed in number of locations on the microtubule.
The $z$ axis represents the probability of having a motor at distance $x$ and a motor a distance $y$ from the rearguard motor. The probabilities of configurations with less than three motors engaged are not reported because that would make the visualization unnecessarily more complicated.
The four snapshots are taken at time $0 sec$, $0.2 sec$, $0.4 sec$, $0.6 sec$ from left to right.
In the second row of Figure~\ref{fig: time histogram 3 motors} we report the results in the case of a high load ($F_{load}=0.008 nN$) with an identical initial probability density function.
Notice how the change in the parameter $F_{load}$ leads to two different dynamics of the probability density function: under a low load the motors tend to spread out while under a high load certain configurations tend to be more likely.

\subsubsection{Average runlength}
The detailed knowledge of the probability vector $P(t)$ can be used to derive information about numerous aspects of the behavior of the motor ensemble and study the variation of certain quantities as a function of different parameters. An important quantity that can also be experimentally measured in experiments is the expected runlength of the motors, that is the average length travelled by the cargo before being lost.
The average runlength is immediately computed from $P(t)$
\begin{align}
	Average~Runlength=E\left[
				\int_{0}^{+\infty} \frac{dx^{(c)}}{dt} dt
			\right]
\end{align}
where $E[\cdot]$ is the mean operator and $x^{(c)}$ is the cargo position.
This same quantity as a function of the load has already been computed in \cite{KunVer08} using Monte Carlo methods in two scenarios. In one scenario (Model A) the authors neglect the effect of thermal noise and in another scenario (Model B) they provide a dynamic model for the brownian motion of the cargo.\\
We report our results in Figure~\ref{fig: runlength} for two reasons.
The first reason is for comparison.
By forcing the variance of the cargo position to zero we obtain a model that is equivalent to ``Model A'' in \cite{KunVer08}. The only difference is that in our case we are computing the average runlength from the probability density function without using a simulation approach.
Thus, forcing the variance of the cargo position to zero, the results must match.
We report the results of our computation in Figure~\ref{fig: runlength}~a using both a coarse grid (solid lines)  and a fine grid (dashed lines)  for the load force.
If we compare the results we computed with the results in \cite{KunVer08} at the same points (that is using the coarse grid), the runlength-load curves overlap almost perfectly.
\begin{figure*}
	\centering
	\begin{tabular}{cc}
		\includegraphics[width=0.45\columnwidth]{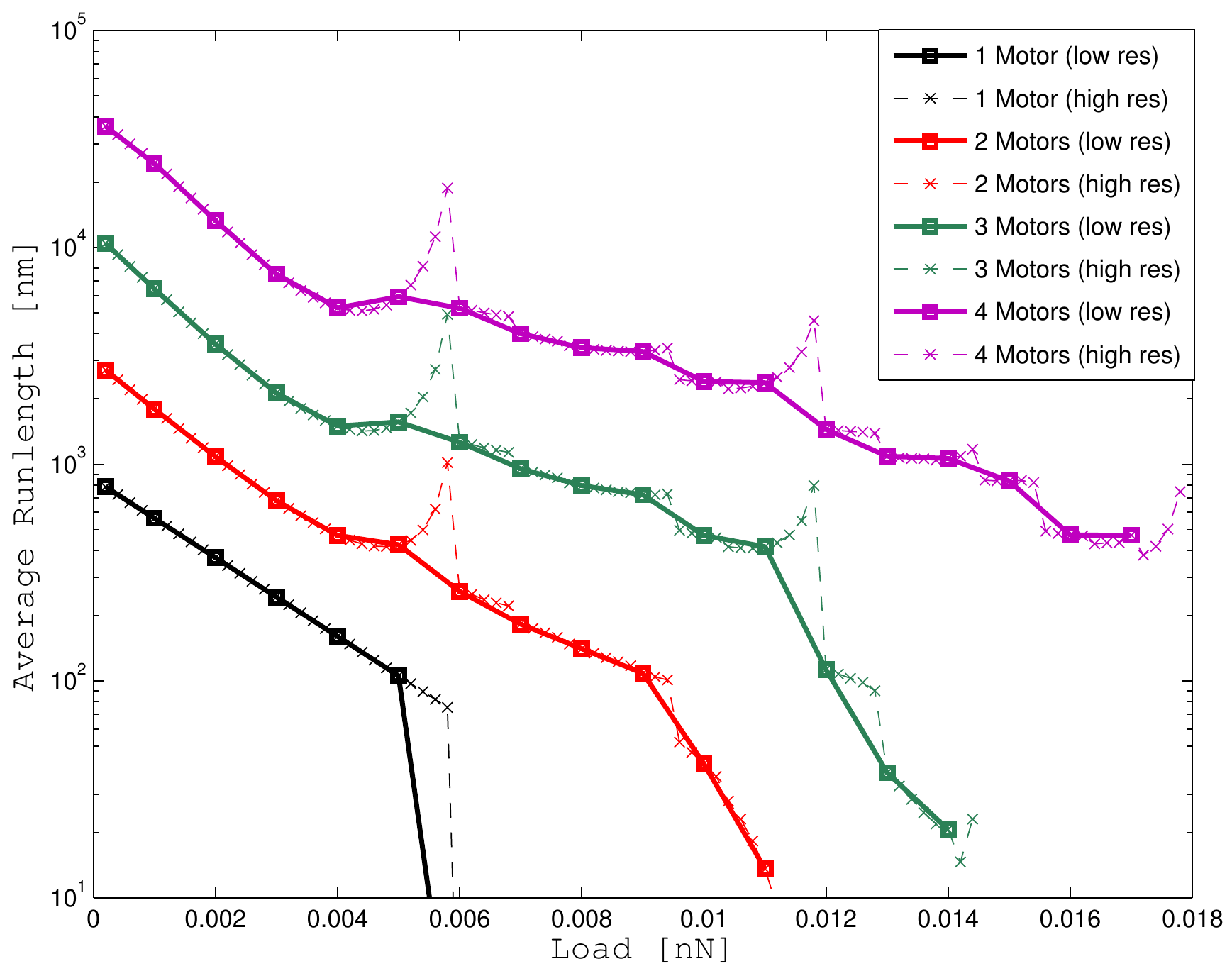}
		\includegraphics[width=0.45\columnwidth]{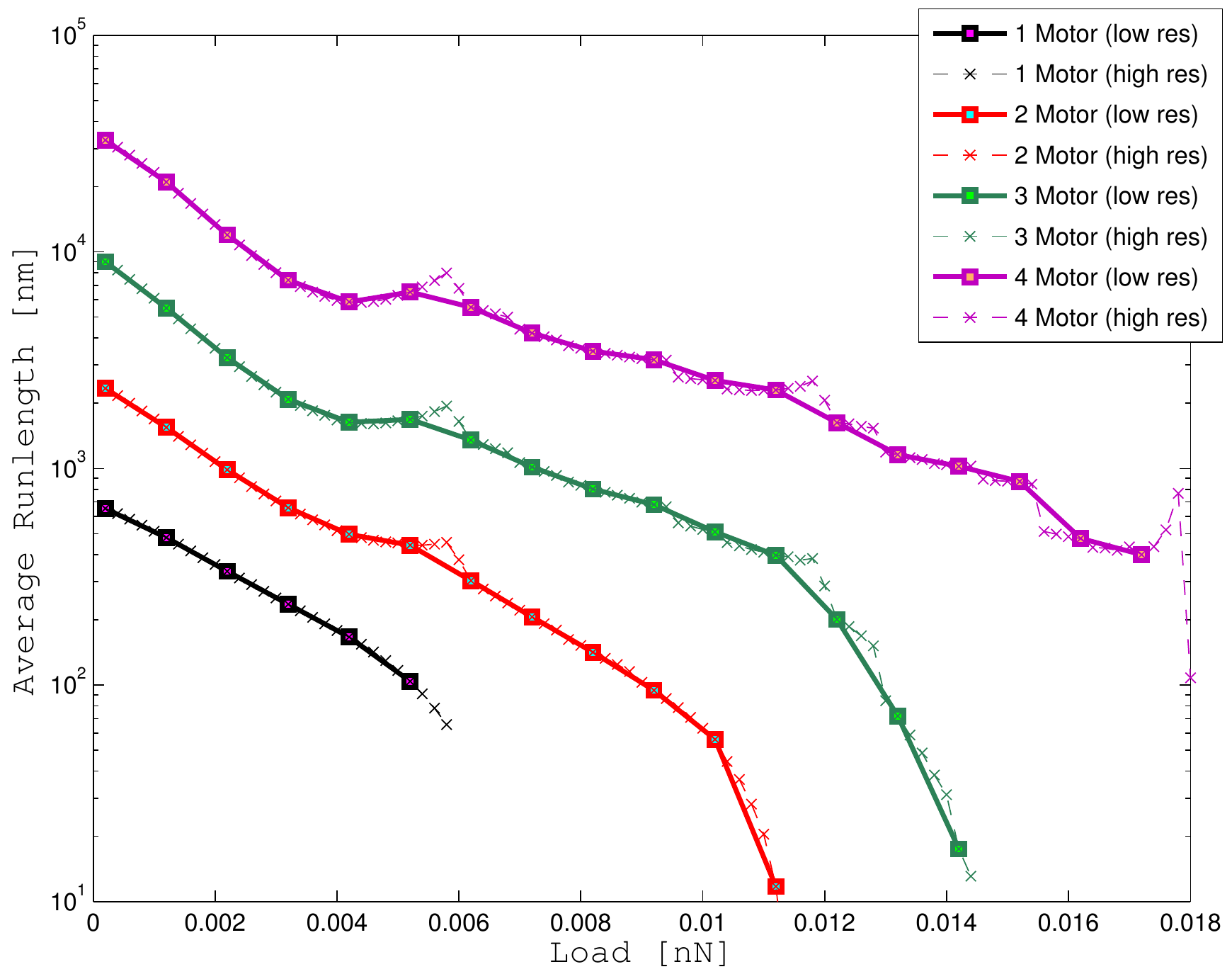}\\
		(a) & (b)
	\end{tabular}
	\caption{Average runlength neglecting thermal noise on the cargo (a) and considering it (b).
	\label{fig: runlength}
	}
\end{figure*}
The second reason is to highlight a possible weakness of the model adopted for the detachment rate.
In the finer scale, we can notice peaks in the runlength curves that were not evident before.
These peaks appear in correspondence of loads that are multiple of the stalling force $F_0$ and can be explained as ``artifacts'' of the model.
They can be explained in the following way.
Let us consider the case $\overline{m}=2$ for simplicity. When only one motor is engaged ($m(t)=1$) and $F_{load}$ is close to (but less than) the stalling force $F_{0}$, the probability of detachment becomes small, and the loss of the cargo becomes unlikely.
In the meantime, the disengaged motor has the opportunity of attaching to microtubule, catching up with the leading motor and moving the cargo a little further.
For $\overline{m}>2$, similar arguments can be used.
This mechanisms explains how, in the mathematical model, by neglecting the brownian motion of the cargo, the expected runlength tends to increase while the load approaches multiple values of the stalling force.\\
When the brownian motion of the cargo is taken into account (see Figure~\ref{fig: runlength}~b) those peaks are smoothed down, but do not disappear completely.
Thus, the pronounced peaks in Figure~\ref{fig: runlength}~a and Figure~\ref{fig: runlength}~b are due to the inaccuracy of the equation (\ref{eq: lambda step}) for $F\simeq F_{s}$.
Still, the insights given by the model could have a physical meaning and could be used to explain possible mild breaks of monotonicity in the runlength curves for multiple motors.

\subsection{Steady state characteristics}

\subsubsection{Average number of engaged motors}
In Figure~\ref{fig: histogram number of motors}~a, the probability distribution $P(t)$ has been used to determine the probability of having a given number of motors engaged on the microtubule at time $t$. 
\begin{figure*}
	\begin{tabular}{cc}
	\includegraphics[width=0.45\columnwidth]{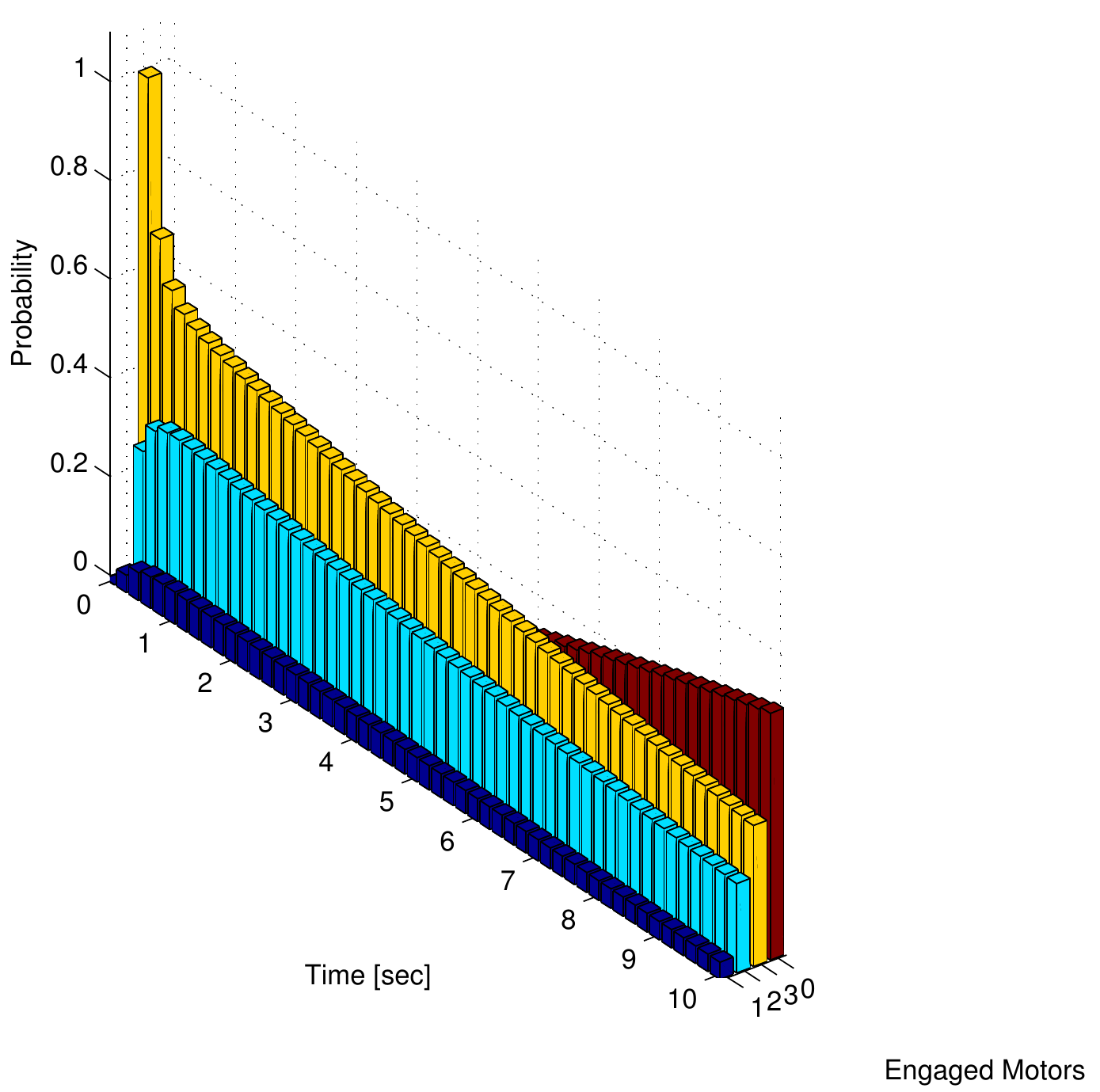} &
	\includegraphics[width=0.45\columnwidth]{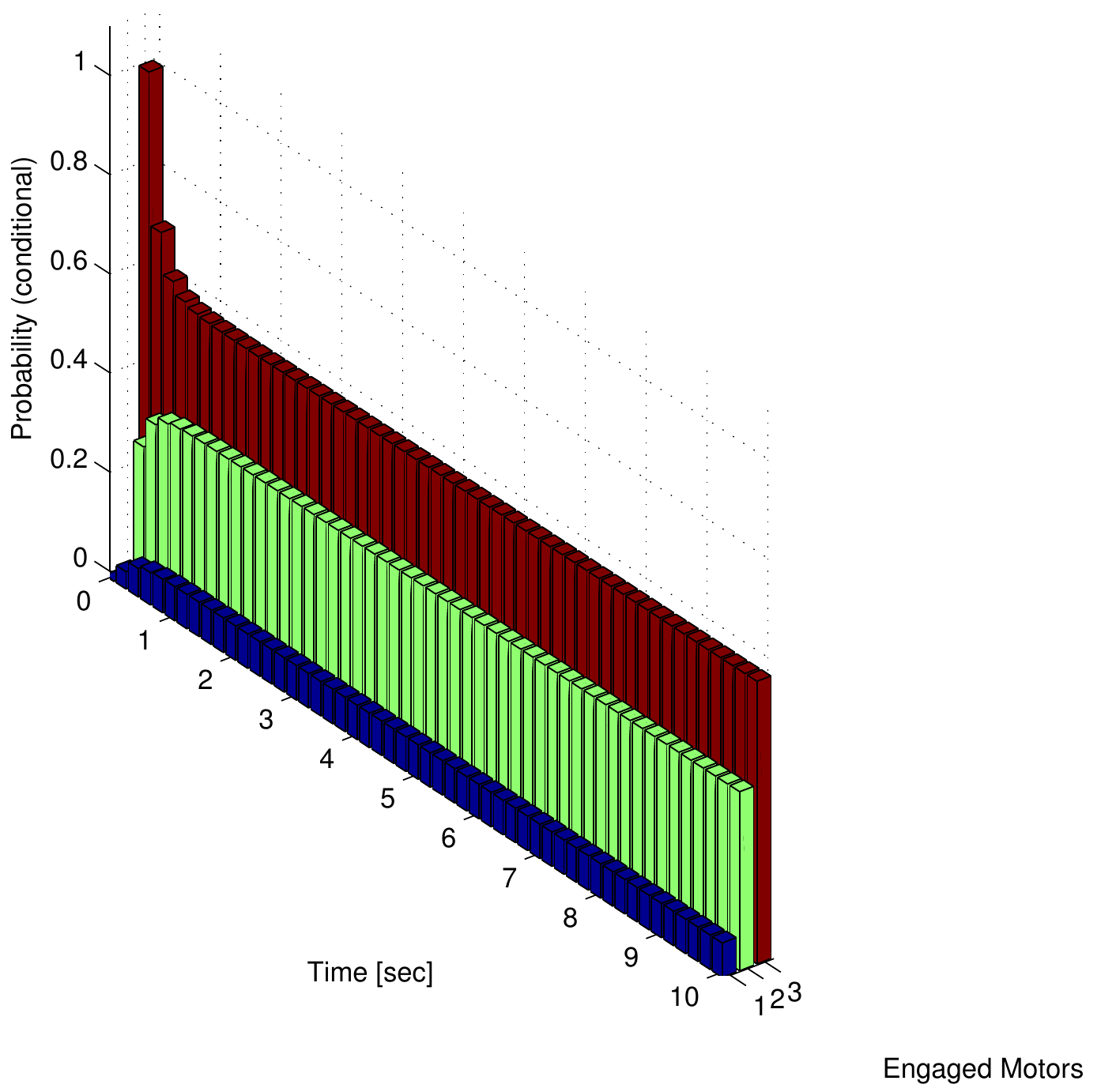}\\
	(a) & (b)
	\end{tabular}
	\caption{Histogram of the number of motors (left) and histogram of the number of motors assuming that the cargo has not been lost, that is $m(t)>0$.
	\label{fig: histogram number of motors}}
\end{figure*}
We notice that the probability of cargo loss increases with time (and it can be proved converges to $1$ for $t\rightarrow +\infty$). Thus, it is not possible to define a steady state for the system: the eventual cargo loss is the actual steady state.
However, in Figure~\ref{fig: histogram number of motors}~b we report the conditional probability density function of the number of motors engaged on the microtubule at time $t$, given that at least one motor is engaged. We observe that using this conditional probability density function a steady state can be defined.
In the case of $\overline{m}=3$, the steady state conditional probability is reported in the histograms of Figure~\ref{fig:steady state cond pdf}~a and in Figure~\ref{fig:steady state cond pdf}~b for $F{load}=0.0002 nN$ and $F{load}=0.008 nN$ respectively.
As in Figure~\ref{fig: time histogram 3 motors}, the rearguard motor is taken as a reference point, the $x$ and $y$ axis represent the distance of the other two motors from the rearguard motor and the $z$ axis represents the probability.
\begin{figure*}
	\begin{tabular}{cc}
	\includegraphics[width=0.45\columnwidth]{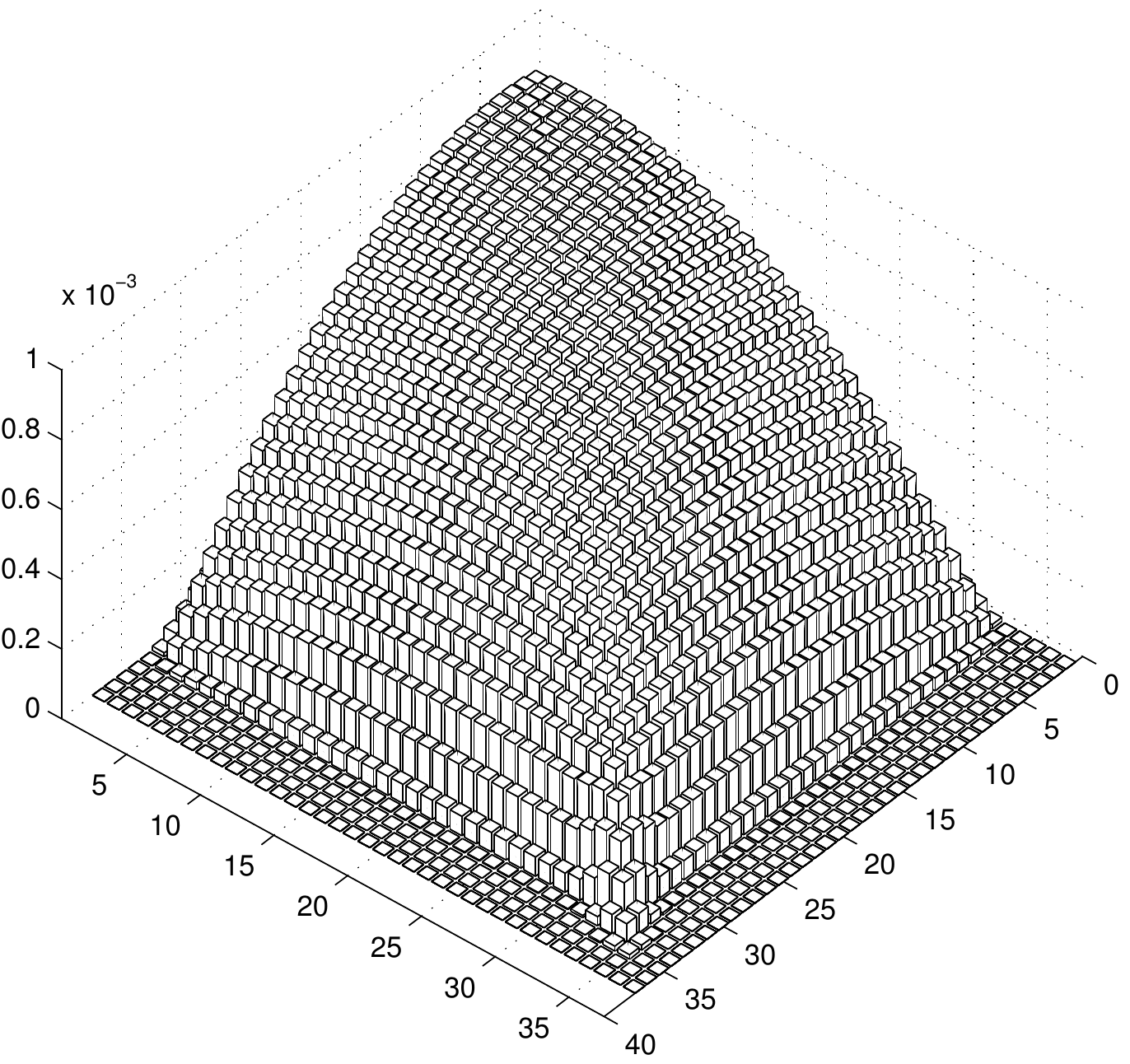} &
	\includegraphics[width=0.45\columnwidth]{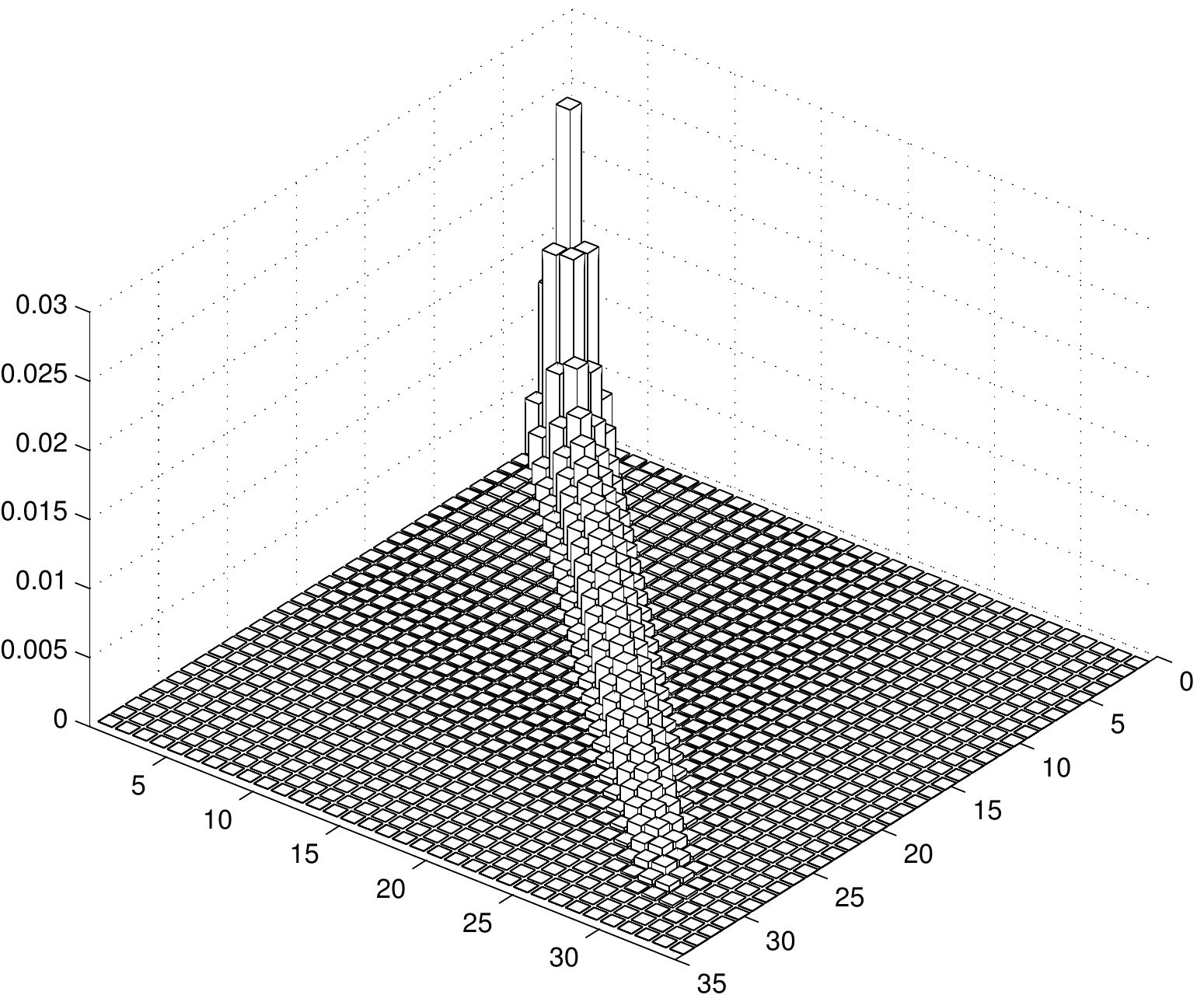}\\
	(a) & (b)
	\end{tabular}
	\caption{
	\label{fig:steady state cond pdf}
	}
\end{figure*}
It can be observed that under a low load the motors tend to be spread and distant from each other, while under a high load they tend to be clusterized. Indeed, the pronounced probability peak around the point $x=0, y=0$ indicates that the three motors are not far from each other. It is interesting to notice that there is a significantly higher probability along the diagonal line $x=y$.
This can be interpreted as follows.
The three motors tend to be close together, but when they are not it is more likely to find the two leading motors together while the rearguard motor is lagging behind.

\subsubsection{Average velocity}
The definition of the steady state conditional probability can be used to determine the average steady state velocity of the motor ensemble.
The results are reported in Figure~\ref{fig: avg vel}~a for the case where thermal fluctuations are neglected and in Figure~\ref{fig: avg vel}~b for the more accurate model.
Results in Figure~\ref{fig: avg vel}~a are again for comparison and verification since, in the noiseless scenario, our model is equivalent to Model A in \cite{KunVer08}.
Results qualitatively and quantitatively non-dissimilar are reported in Figure~\ref{fig: avg vel}~b for the noisy scenario.
\begin{figure*}
	\centering
	\begin{tabular}{cc}
	\includegraphics[width=0.45\columnwidth]{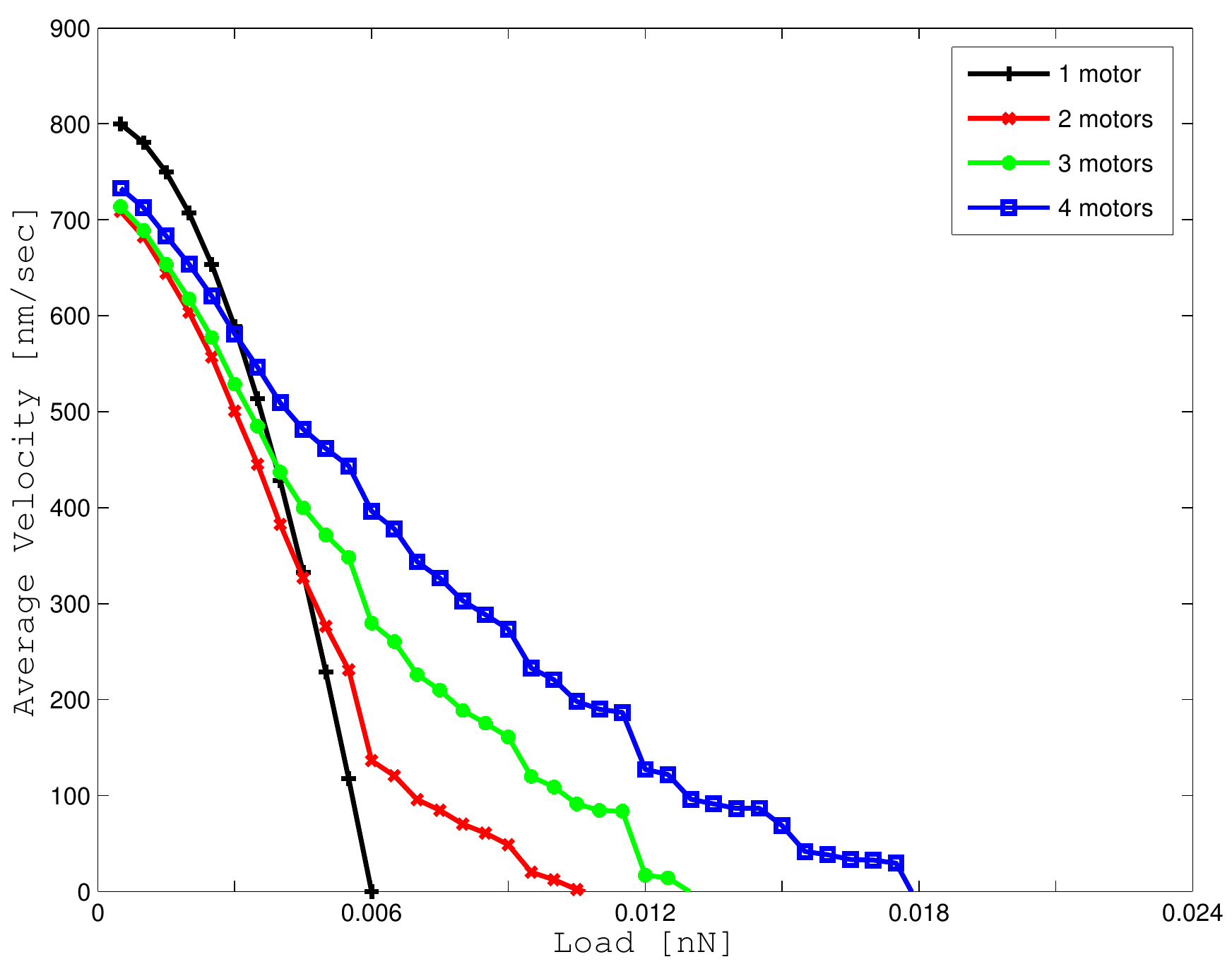} &
	\includegraphics[width=0.45\columnwidth]{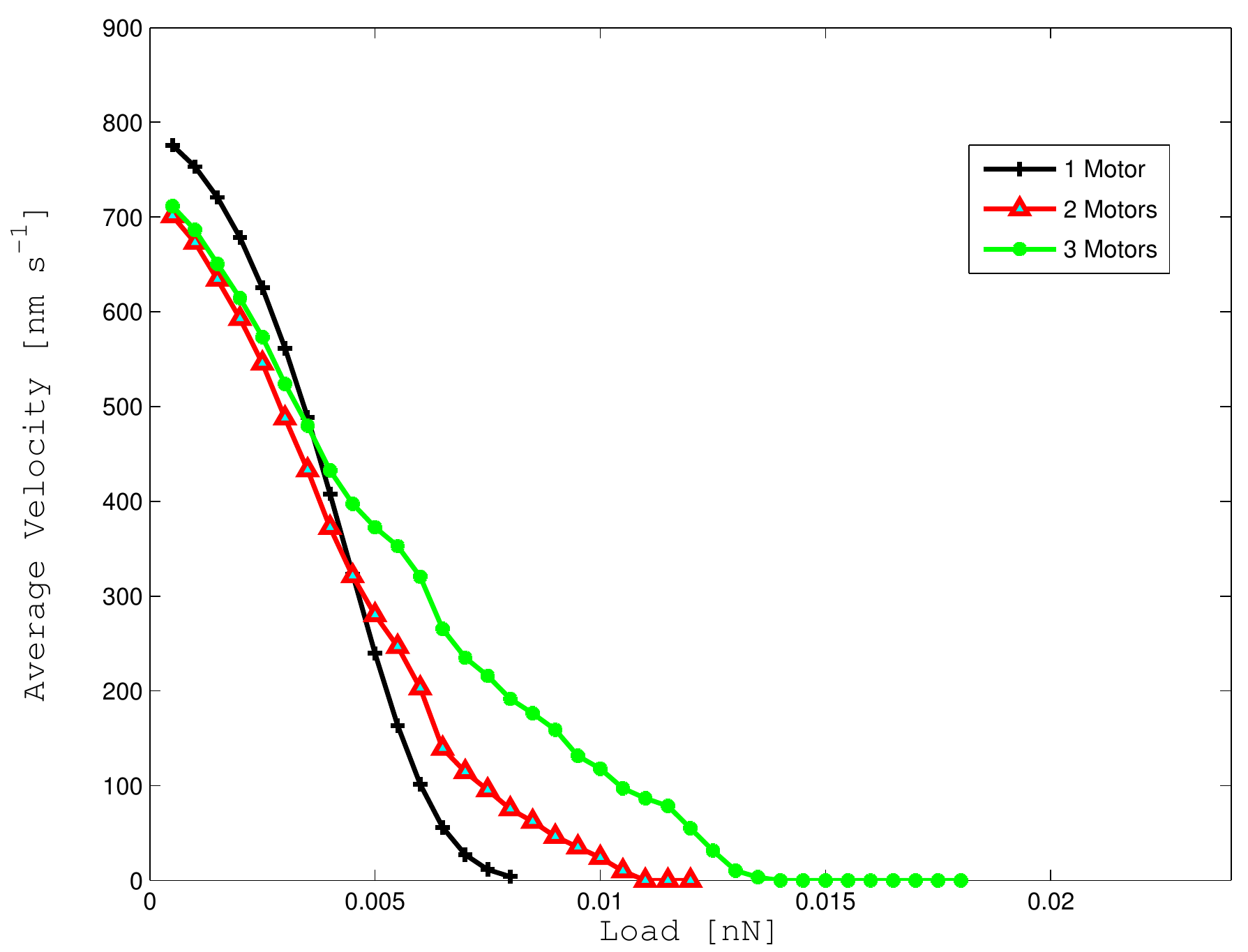}\\
	(a) & (b)
	\end{tabular}
	\caption{
	\label{fig: avg vel}
	}
\end{figure*}

\subsection{Detection of rare events}
The possibility of determining exactly the probability distribution of the different motor configurations allows for the detection of rare events.
It is possible, for example, to determine the probability of the different steps sizes for an ensemble of $2$ motors as represented in Figure~\ref{fig:step histogram 2 motors}.
\begin{figure*}
	\begin{tabular}{cc}
	\includegraphics[width=0.45\columnwidth]{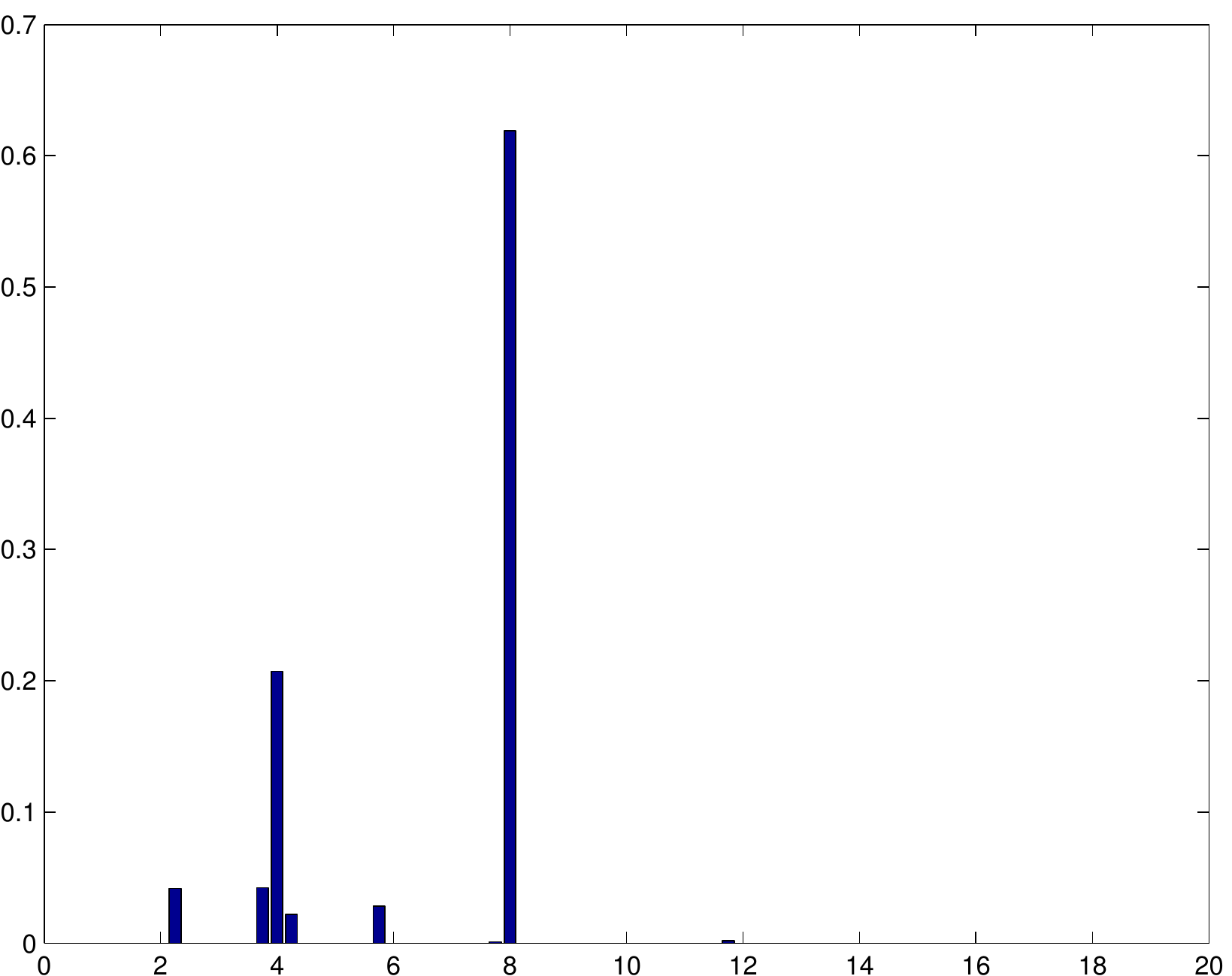} &
	\includegraphics[width=0.45\columnwidth]{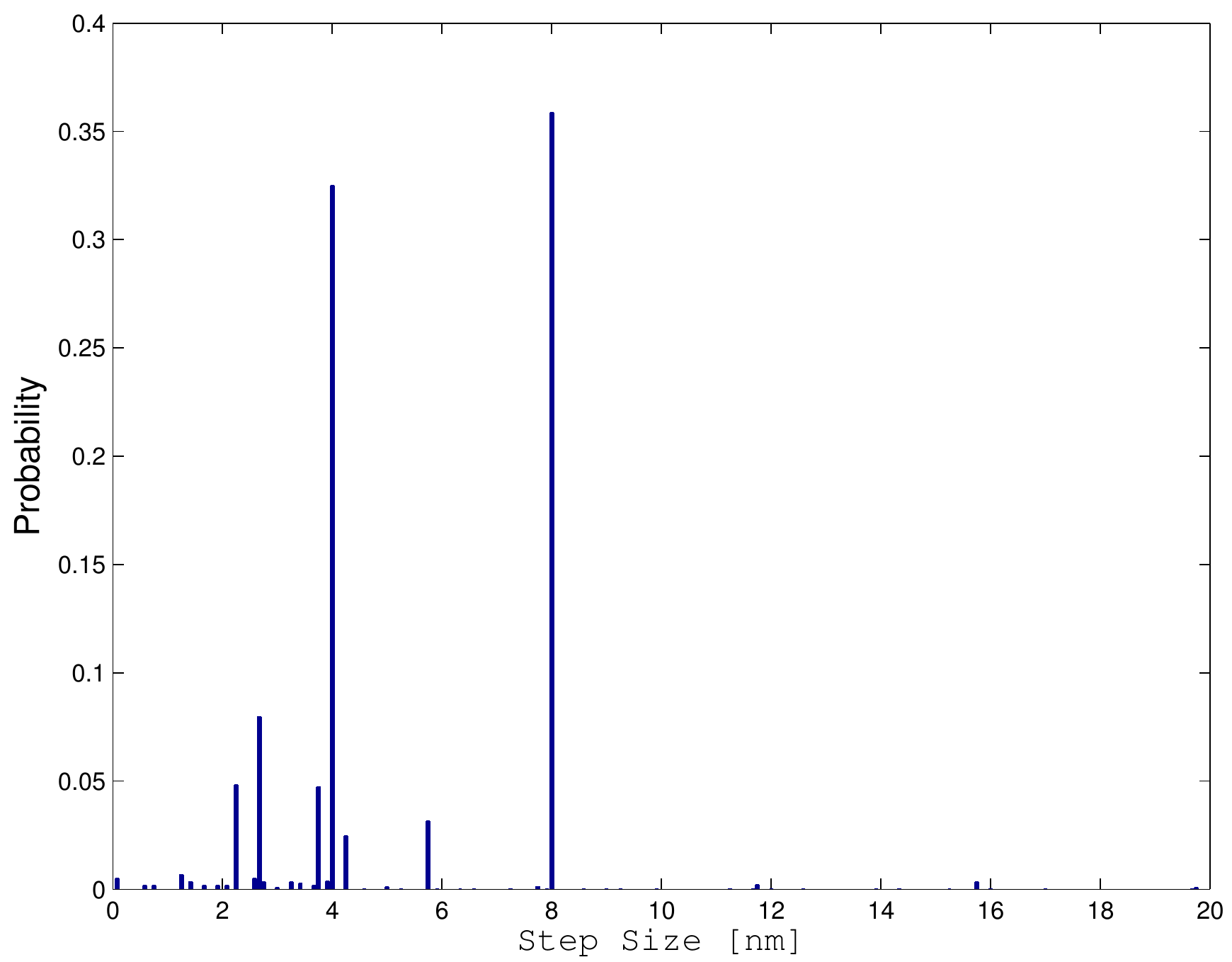}\\
	(a) & (b)
	\end{tabular}
	\caption{
	\label{fig:step histogram 2 motors}
	\label{fig:step histogram 3 motors}
	}
\end{figure*}
As it can be noted, there are two prominent peaks corresponding to $8~nm$ and $4~nm$.
These peaks corresponds respectively to the case where there is one active motor before and after the step and to the case where there are two active motors before and after the step.
There are also different predicted step sizes at about $2~nm$, $6~nm$  and around $4~nm$.
They correspond to situations where there is a different number of active motors before and after the step and they can be mostly considered as artifacts of the model.
However, we also find a small probability, but different from zero, of steps larger than $8~nm$ at about $11~nm$. This does not seem possible, since the each motor can advance only $8~nm$ at the time.
Indeed, this can occur when the rearguard motor is pulling back the cargo and detaches from the microtubule. While this situation is not very likely in a bead-assay experiment that our simulations are trying to portray, it is quite likely that two motors could pull in two opposite directions in a gliding assay experiment. Thus, steps larger that $8~nm$ as those described in
\cite{LedRuh07} could be originated by a mechanism of this kind.
The probability distribution of the step size for an ensemble of $3$ motors is reported in Figure~\ref{fig:step histogram 3 motors} where the steps longer than $8~nm$ have an identical interpretation.

\appendix

\subsection{Preliminary Lemmas}

\begin{lem}\label{lem: unique cargo eq}
	If $F_{load}>0$ and $m>0$, then the equilibrium solution $x^{(c)}$ to (\ref{eq:equilibrium}) for any position sequence $X=(x_1,...,x_m)$ and for $\chi$ as given by (\ref{eq:elastic dead zone}) exists, is unique and
	$x_{m}-l_0>x^{(c)}$.
\end{lem}
\begin{proof}
	The proof is left to the reader.
\end{proof}

\begin{lem}\label{lem:decreasing step}
	Consider an EMM meeting the KVXG rules with parameters $(\overline{m}, K, l_0, F_s)$.
	Assume that $F_{load}>0$, and that, for a configuration $Z$, the number of engaged agents is $m>0$ such that the cargo position $x^{(c)}$ is uniquely defined.
	Also assume $X=\phi_A(Z)=(x_1,...,x_m)$ for a location sequence $A$ and that
	\begin{equation}
		s(Z):=x_m-x_1 \geq  \frac{m F_s-F_{load}}{K}+2l_0.
	\end{equation}
	Then the vanguard motor is stalled, that is $F_s\leq \chi(x_m-x^{(c)})$.
	If the vanguard motor is in the location $a_{k_m}$
	it immediately follows that, for $k\geq k_m$, $\lambda(Z+R^{(s)}(k),Z)=0$.
\end{lem}
\begin{proof}
	Without any loss of generality, assume $x^{(c)}=0$, since the reference point on the microtubule is arbitrary.
	By contradiction, and by using Lemma (\ref{lem: unique cargo eq}), assume $K(x_m-l_0)<F_s$.
	If $x_1<-l_0$, then we have
	\begin{align}
		& F_{load}-K(x_{1}+l_0)= F_{load}+\sum_{i=2}^{m}\chi(x_{i}) \leq \\
		& \leq (m-1)K(x_{m}-l_0) < -K(x_{m}-l_0)+mF_s
	\end{align}
	This implies 
	\begin{align}
		x_{m}-x_{1} < \frac{m F_s-F_{load}}{K}+2 l_0
	\end{align}
	that is a contradiction.
	In the case $x_{1}\geq-l_0$, we have
	\begin{align}
		F_{load} = \sum_{i=1}^{m}\chi(x_{i})\leq mF_s < mF_s+K(x_m-l_0).
	\end{align}
	implying
	\begin{align}
		x_{m}-x_{1}\leq x_m+l_0 < \frac{m F_s-F_{load}}{K}+2 l_0,
	\end{align}
	that is again a contradiction.
\end{proof}

\begin{lem}\label{lem:decreasing attachment}
	Consider an EMM meeting the KVXG rules with parameters $(\overline{m}, K, l_0, F_s)$.
	Assume that $F_{load}>0$, and that, for a configuration $Z$, the number of engaged agents is $m>0$ such that the cargo position $x^{(c)}$ is uniquely defined.
	Assume that $s(Z)=x_{m}-x_1\leq q$,
	with 
	\begin{align}
		q\geq\frac{F_{load}}{K}+2l_0,
	\end{align}
	and that the next transition occurring is an attachment at position $k$, that is $Z'=Z+R^{(a)}(k)$.
	Then, if $\lambda(Z',Z)>0$
	\begin{align}
		s(Z')\leq q.
	\end{align}
\end{lem}
\begin{proof}
	Without any loss of generality, assume $x^{(c)}=0$, since the reference point on the microtubule is arbitrary.
	Let $X'=(x'_1, ..., x_{m+1}')=\phi_{A}(Z')$
	If $x_1<-l_0$, the transition does not affect the value of the distance because the newly attached motor will be in the range $[-l_0, l_0]$ in order to meet the KVXG rules. Indeed, we have 
	\begin{align}
		s(Z')=x_{m+1}'-x_{1}'=x_{m}-x_{1}=s(Z) \leq q
	\end{align}
	If $x_1\geq-l_0$, then $x_1'\geq-l_0$. Moreover,
	\begin{align}
		F_{load}=\sum_{i=1}^{m}\chi(x_i)\geq K(x_{m}-l_0),
	\end{align}
	implies
	\begin{align}
		x_{m+1}'-x_1'\leq x_{m}+l_0 \leq \frac{F_{load}}{K}+2l_0 \leq q.
	\end{align}
\end{proof}

\subsection{Proof of Theorem \ref{thm:bounded configuration}}
\begin{proof}
	Define, for $t\geq t_0$,
	\begin{align}
		& 0\leq \phi(t):=Pr\{s(t)> S| Z(t_0)=\overline{Z}\}=\\
			&\qquad =\sum_{s(Z)>S}P_{Z}(Z,t|\overline{Z},t_0).
	\end{align}
	At time $t_0$ we have that $\phi(t_0)=0$.
	By contradiction assume that there exists a time $t_1>t_0$ such that
	$\phi(t_1)>0$.
	We must also have that, for some $t\in(t_0,t_1)$,
	\begin{align}
		\frac{\partial}{\partial t}\phi(t)>0.
	\end{align}
	From the master equation
	\begin{align*}
		&\frac{\partial}{\partial t}\phi(t)=
			\sum_{s(Z)>S} \left\{-\sum_{Z'\in\mathcal{Z}}\lambda(Z',Z)P_{Z}(Z,t|\overline{Z},t_0)\right.+\\
		&\qquad \left. +\sum_{Z'\in\mathcal{Z}}\lambda(Z,Z')
			P_{Z}(Z',t|\overline{Z},t_0)\right\}=\\
		&=\sum_{s(Z)>S} \sum_{s(Z')\leq S}\left\{-\lambda(Z',Z)P_{Z}(Z,t|\overline{Z},t_0)\right.+\\
		&\qquad \left.+\lambda(Z,Z')
			P_{Z}(Z',t|\overline{Z},t_0)\right\}\leq\\
		&\leq \sum_{s(Z)>S} \sum_{s(Z')\leq S}
			\left\{
				\lambda(Z,Z')P_{Z}(Z',t|\overline{Z},t_0)
			\right\}
	\end{align*}
	Now we prove that, if $s(Z')\leq S$ and $s(Z)>S$, we necessarily have $\lambda(Z,Z')=0$.\\
	The only possible transitions are detachment, attachment and step.\\
	The detachment transition can only reduce the configuration distance, so $Z-Z'\neq R^{(d)}(k)$ for any $k$.\\
	If the transition is an attachment, by Lemma \ref{lem:decreasing attachment} we obtain the assertion.\\
	If the transition is a step the configuration distance increases only if the vanguard motor steps, and the increase is by $d_{s}$. If $s(Z')< S-d_{s}$,
	then $s(Z)< S$ and this is not a case we are considering because we have assumed $s(Z)>S$. Conversely, if $s(Z')\geq S-d_{s}$, by Lemma \ref{lem:decreasing step}, we have that the vanguard motor can not step, so such a transition is not possible. Then we obtain as a contradiction that, for $t\geq t_0$,
	\begin{align}
		\frac{\partial}{\partial t}\phi(t)\leq 0.
	\end{align}
\end{proof}


\bibliographystyle{IEEEtran}
\bibliography{../../../bibliography/motors.bib}

\begin{thebibliography}{10}
\providecommand{\url}[1]{#1}
\csname url@rmstyle\endcsname
\providecommand{\newblock}{\relax}
\providecommand{\bibinfo}[2]{#2}
\providecommand\BIBentrySTDinterwordspacing{\spaceskip=0pt\relax}
\providecommand\BIBentryALTinterwordstretchfactor{4}
\providecommand\BIBentryALTinterwordspacing{\spaceskip=\fontdimen2\font plus
\BIBentryALTinterwordstretchfactor\fontdimen3\font minus
  \fontdimen4\font\relax}
\providecommand\BIBforeignlanguage[2]{{%
\expandafter\ifx\csname l@#1\endcsname\relax
\typeout{** WARNING: IEEEtran.bst: No hyphenation pattern has been}%
\typeout{** loaded for the language `#1'. Using the pattern for}%
\typeout{** the default language instead.}%
\else
\language=\csname l@#1\endcsname
\fi
#2}}

\bibitem{CarCro2005}
N.~J. Carter and R.~Cross, ``{Mechanics of the kinesin step},'' \emph{Nature},
  vol. 435, no. 7040, p. 308, 2005.

\bibitem{SvoSch93}
K.~Svoboda, C.~Schmidt, B.~Schnapp, and S.~Block, ``{Direct observation of
  kinesin stepping by optical trapping interferometry},'' \emph{Nature}, vol.
  365, no. 6448, pp. 721--727, 1993.

\bibitem{MilYil06a}
L.~Milescu, A.~Yildiz, P.~Selvin, and F.~Sachs, ``{Maximum likelihood
  estimation of molecular motor kinetics from staircase dwell-time
  sequences},'' \emph{Biophysical journal}, vol.~91, no.~4, pp. 1156--1168,
  2006.

\bibitem{CarVer08}
B.~Carter, M.~Vershinin, and S.~Gross, ``{A Comparison of Step-Detection
  Methods: How Well Can You Do?}'' \emph{Biophysical Journal}, vol.~94, no.~1,
  pp. 306--319, 2008.

\bibitem{LieLip09}
S.~Liepelt and R.~Lipowsky, ``{Operation modes of the molecular motor
  kinesin},'' \emph{Phys. Rev. E}, vol.~79, p. 011917, 2009.

\bibitem{KunVer08}
A.~Kunwar, M.~Vershinin, J.~Xu, and S.~P. Gross, ``Stepping, strain gating, and
  an unexpected force-velocity curve for multiple-motor-based transport,''
  \emph{Current biology}, vol.~18, p. 1173, 2008.

\bibitem{KurKim05}
C.~Kural, H.~Kim, S.~Syed, G.~Goshima, V.~I. Gelfand, and P.~R. Selvin,
  ``{Kinesin and dynein move a peroxisome in vivo: a tug-of-war or coordinated
  movement?}'' \emph{Science}, vol. 308, no. 5727, pp. 1469--1472, 2005.

\bibitem{SopRai09}
V.~Soppina, A.~K. Rai, A.~J. Ramaiya, P.~Barak, and R.~Mallik, ``{Tug-of-war
  between dissimilar teams of microtubule motors regulates transport and
  fission of endosomes},'' \emph{PNAS}, vol. 106, no.~46, pp. 19\,381--19\,386,
  2009.

\bibitem{MeyHow95}
E.~Meyh{\"o}fer and J.~Howard, ``{The force generated by a single kinesin
  molecule against an elastic load},'' \emph{Proceedings of the National
  Academy of Sciences of the United States of America}, vol.~92, no.~2, p. 574,
  1995.

\bibitem{LedRuh07}
C.~Leduc, F.~Ruhnow, J.~Howard, and S.~Diez, ``{Detection of fractional steps
  in cargo movement by the collective operation of kinesin-1 motors},''
  \emph{Proceedings of the National Academy of Sciences}, vol. 104, no.~26, p.
  10847, 2007.

\end{thebibliography}

\end{document}